\documentclass[journal]{IEEEtran}

\usepackage{bm}
\usepackage{amsfonts,dsfont}
\usepackage{balance}
\usepackage{booktabs}
\usepackage{amsmath}
\usepackage{tabularx}
\usepackage{subfigure}
\usepackage{graphicx}
\usepackage{multirow}
\usepackage{rotating}
\usepackage{algorithm}
\usepackage{algorithmic}
\usepackage{url}
\usepackage{amssymb}
\usepackage{amsthm}
\usepackage{enumerate}
\usepackage{epstopdf}
\usepackage{siunitx}
\usepackage{collcell}
\usepackage{booktabs}
\usepackage{diagbox}

\usepackage{color}
\newcommand{\green}[1]{\textcolor{green}{#1}}

\newcommand{\arrowednumber}[2]{%
  \ifnum#2 >= 0
    #1\textsuperscript{$\uparrow$}%
  \else
    #1\textsuperscript{$\downarrow$}%
  \fi
}

\newtheorem{theorem}{Theorem}[section]

\begin{document}
%
\title{Can Encrypted Images Still Train Neural Networks? Investigating Image Information and Random Vortex Transformation}

\author{Xiao-Kai~Cao,~
 Wen-Jin~Mo,~
 Chang-Dong~Wang,~\IEEEmembership{Senior Member,~IEEE,}
  Jian-Huang~Lai,~\IEEEmembership{Senior Member,~IEEE,}\\
 Qiong~Huang,~\IEEEmembership{Member,~IEEE}
\thanks{The work was supported by National Key Research and Development Program of China (2021YFF1201200) and NSFC (62276277).}
\thanks{Xiao-Kai Cao, Wen-Jin~Mo, Chang-Dong Wang and Jian-Huang Lai are with School of Computer Science and Engineering, Sun Yat-sen University, Guangzhou, China, and also with Key Laboratory of Machine Intelligence and Advanced Computing, Ministry of Education, China.
  E-mail: caoxk@mail2.sysu.edu.cn, mowj27@mail2.sysu.edu.cn, changdongwang@hotmail.com, stsljh@mail.sysu.edu.cn.}
  \thanks{Qiong Huang is with College of Mathematics and Informatics, South China Agricultural University, Guangzhou, China.
  E-mail: qhuang@scau.edu.cn.}
  \thanks{Corresponding author: Chang-Dong Wang.}
  }

\markboth{IEEE Transactions on Image Processing}%
{\MakeLowercase{\textit{Cao et al.}}: }

\maketitle

\begin{abstract}
Vision is one of the essential sources through which humans acquire information. To simulate this biological characteristic, researchers have proposed models such as Convolutional Neural Networks and Vision Transformers to mimic the human visual system's processing of visual tasks like image classification. These models have achieved significant success in experiments. However, there has been limited theoretical analysis of the variation in information content during image processing. This is not an insignificant task, as it plays a crucial guiding role in image feature extraction, image quality assessment, and the design of image encryption algorithms. In this paper, we establish a novel framework for measuring image information content to evaluate the variation in information content during image transformations. Within this framework, we design a nonlinear function to calculate the neighboring information content of pixels at different distances, and then use this information to measure the overall information content of the image. Hence, we define a function to represent the variation in information content during image transformations. Additionally, we utilize this framework to prove the conclusion that swapping the positions of any two pixels reduces the image's information content. Furthermore, based on the aforementioned framework, we propose a novel image encryption algorithm called Random Vortex Transformation. This algorithm encrypts the image using random functions while preserving the neighboring information of the pixels. The encrypted images are difficult for the human eye to distinguish, yet they allow for direct training of the encrypted images using machine learning methods. Experimental verification demonstrates that training on the encrypted dataset using ResNet and Vision Transformers only results in a decrease in accuracy ranging from 0.3\% to 6.5\% compared to the original data, while ensuring the security of the data. Furthermore, there is a positive correlation between the rate of information loss in the images and the rate of accuracy loss, further supporting the validity of the proposed image information content measurement framework. The experimental code is available for download on \url{https://github.com/CaoXiaokai/Random_Vortex_Transformation}.
\end{abstract}


%
\IEEEpeerreviewmaketitle

\section{Introduction}\label{sec:introduction}
\IEEEPARstart{T}{he} visual system plays a crucial role in our daily lives, as most individuals perceive the external world and gather a wealth of information through visual perception. In order to simulate and understand the human visual system, researchers have delved into the field of computer vision. Computer vision utilizes computer algorithms and techniques to process and analyze image and video data, extracting information about target objects, scenes, and other features. It aims to mimic human-like perception, understanding, and interpretation of visual data. This field has experienced rapid development over the past two decades. In 1998, LeCun et al. provided a detailed description of the Convolutional Neural Network (CNN) model, known as LeNet-5, for handwritten digit recognition. They introduced concepts such as convolutional layers, pooling layers, and fully connected layers, which have since been widely adopted~\cite{DBLP:journals/pieee/LeCunBBH98}. CNNs are characterized by their ability to extract features from images by considering the relationships between neighboring pixels. With the development of deep neural networks, in 2012, Krizhevsky et al. proposed a deep CNN called AlexNet. They improved the model by incorporating techniques such as the ReLU activation function and Dropout~\cite{DBLP:conf/nips/KrizhevskySH12}. In 2014, Simonyan and Zisserman introduced VGGNet, a deeper CNN architecture that utilized smaller convolutional kernels and increased network depth to enhance the model's expressive power~\cite{DBLP:journals/corr/SimonyanZ14a}. In 2016, He et al. addressed issues like model degradation in deep network structures and proposed the Residual Learning (ResNet) model, which leverages a feedforward low-level feature to deeper layers~\cite{DBLP:conf/cvpr/HeZRS16}. In fact, the aforementioned convolutional neural networks or residual neural networks all extract image features by learning the relationships between neighboring pixels.

In order to better capture the relationships between different positions in an image, Dosovitskiy proposed a CNN-independent Vision Transformer in 2021~\cite{DBLP:conf/iclr/DosovitskiyB0WZ21}. This is a variant of the neural network model based on attention mechanisms known as Transformer~\cite{DBLP:conf/nips/VaswaniSPUJGKP17}. Furthermore, the Transformer has been extensively applied in the field of computer vision. In 2020, Carion et al. introduced the Detection Transformer (DETR) for object detection~\cite{DBLP:conf/eccv/CarionMSUKZ20}. In 2021, Liu et al. proposed a new visual transformer called Swin Transformer based on Shifted Windows, which has been widely used in tasks such as image classification, object detection, and semantic segmentation~\cite{DBLP:conf/iccv/LiuL00W0LG21}. When processing images, the Transformer captures the relationships between different positions through self-attention mechanisms, enabling it to learn global information. This allows the Transformer to exhibit stronger performance than ResNet in many computer vision tasks.

These models extract image features by learning the relationships between pixels and have achieved significant success in experiments. However, there is limited literature that theoretically analyzes the relationships between pixels and studies the information contained in images based on these relationships~\cite{DBLP:journals/pami/Marin-FranchF13}. In fact, there is profound research value in studying and analyzing the information of images based on the relationships between pixels. By researching and analyzing the information of images, theoretical support can be provided for image-related techniques and offer guidance for their development. For instance, it plays a role in evaluating and measuring image quality~\cite{DBLP:journals/tip/LiLLLLX23, DBLP:journals/tip/WangSXYWL23, DBLP:journals/tip/MadhusudanaBWAB23}, quantifying the degree of image distortion~\cite{DBLP:journals/tip/YangYKCGS23, DBLP:journals/tip/YanGW0ZS23}, evaluating the effectiveness of image compression and storage algorithms~\cite{DBLP:journals/tip/SunFZ23}, assisting in feature extraction~\cite{DBLP:journals/tip/LiuZXFYW23}, similarity measurement, performance evaluation of image encryption algorithms, and more, as shown in \figurename~\ref{Researching and analyzing the application scenarios of image information}.
\begin{figure*}[!t]%
  \centering
  \includegraphics[width=0.8\textwidth]{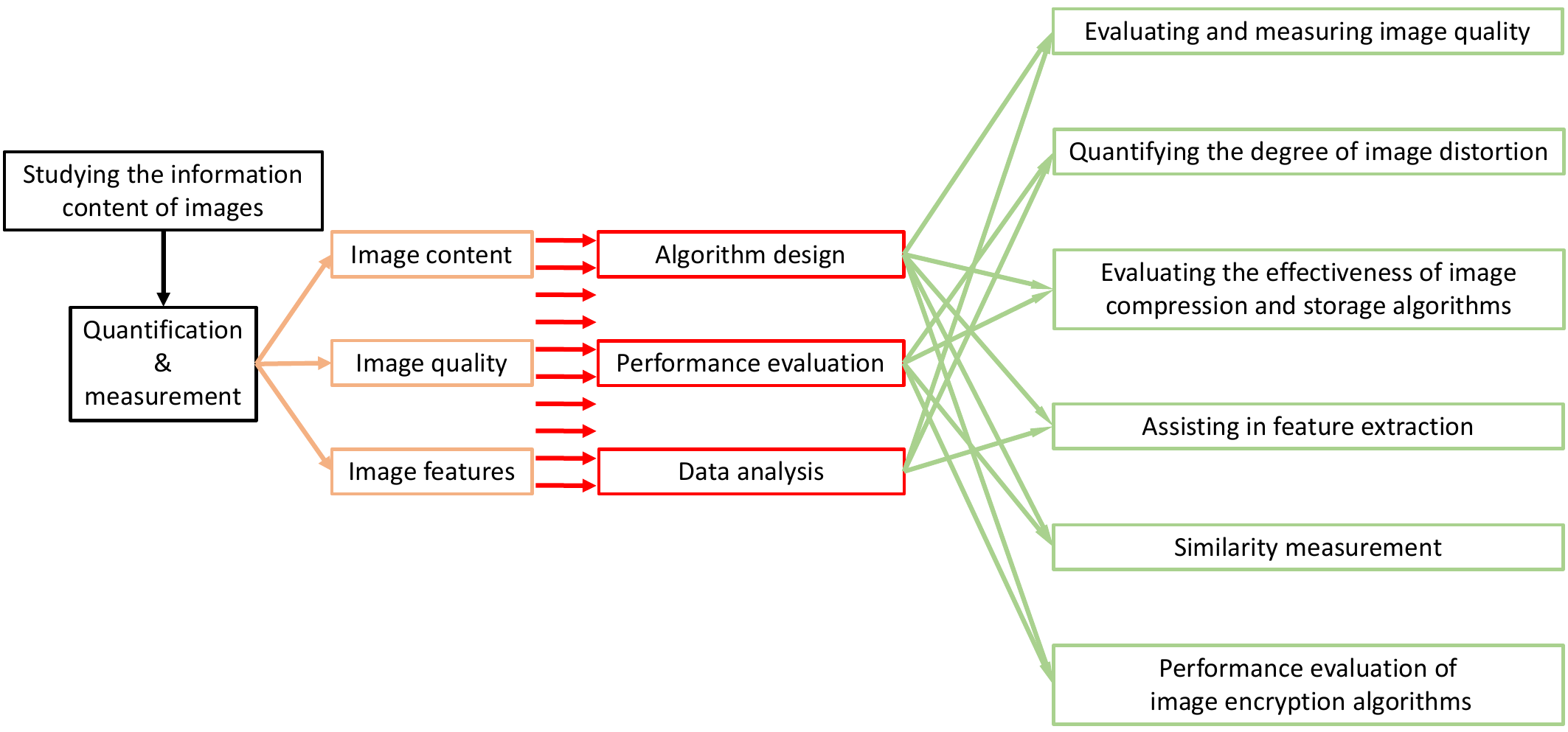}\\
  \caption{Researching and analyzing the application scenarios of image information.}\label{Researching and analyzing the application scenarios of image information}
\end{figure*}

Early image encryption methods often employed symmetric encryption algorithms such as DES and AES, which operated at the pixel level but lacked protection for image structure and semantics~\cite{bruce1996applied, schneier2007applied}. In 1990, Pecora et al. proposed an image encryption method based on chaos theory. They utilized random numbers generated by a chaotic sequence generator for image encryption, thereby providing higher encryption strength and randomness~\cite{pecora1990synchronization}. However, increased randomness could lead to information loss during the encryption and decryption processes. To address this, reversible image encryption algorithms were introduced and widely applied. These algorithms typically involve operations such as pixel position permutation and image diffusion, such as Arnold transform and Discrete Wavelet Transform~\cite{DBLP:conf/mmsec/FridrichGD01}. In recent years, with the advancement of deep learning, more research has emerged exploring the application of deep learning in the field of image encryption. These methods utilize neural network models for image encryption and decryption, offering higher security and resistance against attacks~\cite{DBLP:journals/iotj/DingWCZGCQ21}.

Although these encryption schemes have different principles, most of them serve data transmission purposes. The basic steps involve the data provider encrypting the image using a key, transmitting the encrypted data to the data receiver, and the receiver decrypting the image using a private key. Finally, the data receiver performs operations such as data analysis on the data. However, in reality, there are cases where the data provider does not want the data receiver to obtain the original image as it would leak information about the original image. Is it possible to perform data analysis without decrypting the data? Homomorphic encryption can partially achieve this functionality, but it is only applicable to homomorphic addition and multiplication and is not suitable for commonly used machine learning models.

To address the aforementioned issues, we first define an information content function for images based on the pixel values and the neighboring relationships between pixels. Inspired by CNN and Vision Transformer, we expand the definition of neighboring relationships to include not only adjacent pixels but also the global relationships between each pixel and all other pixels. We design a nonlinear information content function that captures these relationships effectively. Building upon this definition, we conduct a theoretical analysis of the information content changes after image transformations. To the best of our knowledge, there is currently no existing literature that has provided a theoretical analysis of this aspect. Furthermore, to achieve the goal of ``performing data analysis on images without decryption", we propose a technique called Random Vortex Transformation, which is a randomized function-based image encryption method. By applying multiple random vortex transformations to distort the image, we can encrypt the data while preserving the neighboring relationships between pixels, making it difficult for human observers to discern the target information within the image. Moreover, since the random vortex transformation retains the neighboring information of pixels, machine learning methods can be employed for image recognition. To validate the effectiveness of our encryption scheme, we train advanced models such as ResNet and Vision Transformer directly on the encrypted data. The results demonstrate that the accuracy difference between models trained on the encrypted data and the original data is less than $2.5\%$ (for single-channel datasets) and $6.5\%$ (for three-channel datasets). The specific contributions of this paper are as follows:
\begin{itemize}
\item We define an information content function for images to characterize the changes in information content after various transformations are applied to the images.

\item The conclusion that ``exchanging any two pixels will decrease the information content of the image" is theoretically proven.

\item We propose a random vortex transformation for image encryption, which encrypts the data while preserving the neighboring information of pixels, making it difficult for human observers to discern the target information within the image. However, computers can still recognize the content, enabling data analysis under the premise of privacy protection.

\item The effectiveness of the random vortex transformation is verified through experiments. A comparison is made between the random permutation and the random vortex transformation datasets. The results show that random permutation severely disrupts the image information, while the random vortex transformation achieves image encryption with minimal loss of information content.
\end{itemize}

The rest of this paper is organized as follows. In section~\ref{sec: related work}, we provide a brief overview of methods for measuring the information content of images and image encryption techniques. In section~\ref{sec: Information Content of Images}, we define the information content of an image and provide a proof for the ``Principle of Verisimilitude". In section~\ref{sec: Random Vortex Transformation}, we introduce the random vortex transformation proposed in this paper and provide detailed explanations of each component. In section~\ref{sec: Experiments}, we verify the effectiveness of the algorithm on three datasets. Finally, we conclude this paper in section~\ref{sec: Conclusion}.

\section{Related Work}\label{sec: related work}
In this section, we introduce the current mainstream methods for quantifying image information and image encryption techniques.

The information content of an image is typically determined by the pixel values and the neighboring information of adjacent pixels. For images, the neighboring relationships between pixels often contain crucial information regarding the contours of objects within the image. As a result, researchers commonly analyze the neighboring information. The prevailing approach involves calculating the correlation coefficient between adjacent pixels. Specifically, if the correlation coefficient between each pixel and its neighboring pixels in the horizontal, vertical, and diagonal directions tend to approach $1$, it indicates positive correlation among these pixels; otherwise, it signifies negative correlation~\cite{DBLP:journals/tie/ZhangLZ0ZP22, rani2022grayscale}. Correlation coefficients are utilized to measure the degree of correlation between the encrypted image and the original image, thereby evaluating the performance of encryption algorithms. However, such methods only consider the relationships between neighboring pixels, while each object in an image is typically composed of hundreds or even thousands of pixels. We provide a detailed explanation of this in Section \ref{subsec: Definition} of this paper.

Another way to measure the information content of an image is by using the Fr\'{e}chet distance to measure the similarity between two ordered curves~\cite{DBLP:conf/wads/DriemelP21, DBLP:conf/soda/BringmannDNP22, DBLP:conf/soda/ChengH23}. Researchers have considered the contours of objects in an image as curves and used the Fr\'{e}chet distance to precisely learn the similarity between these contours, enabling the identification of objects in the image~\cite{DBLP:conf/icml/ZhangGCDZY23, DBLP:conf/cvpr/Parmar0Z22, DBLP:conf/cvpr/RamtoulaG0M23}. While this method is highly accurate, it is also sensitive to noise. Additionally, although these methods consider factors such as the shape, length, and direction of the curve contours, they overlook the contextual information between curves. For example, they may accurately identify a nose but fail to capture additional information such as the eyes and mouth around the nose.

In the context of image encryption techniques, one of the current mainstream encryption schemes involves utilizing chaotic systems for image encryption~\cite{DBLP:journals/eswa/BhatM22, xu2022new}. Chaotic systems are characterized by the phenomenon that tiny variations in initial conditions can lead to significant changes in system behavior. Exploiting this characteristic of chaotic systems~\cite{DBLP:journals/isci/ZhouWZ23}, pseudo-random sequences can be generated as encryption keys for image encryption. Subsequently, the parameters for encryption techniques such as image diffusion and pixel permutation are generated using the encryption key, and based on these parameters, the image is encrypted. The encrypted data is then transmitted from the data provider to the data receiver, and finally decrypted by the data receiver. One of the objectives of these image encryption methods is to ensure the security of data during transmission. However, does the data receiver obtaining decrypted data satisfy privacy protection requirements? Is it feasible to expect the data receiver to use the data without recovering the original data? Homomorphic encryption does indeed offer similar functionality~\cite{rivest1978data, gentry2009fully}, but the data encrypted using homomorphic encryption is only applicable for homomorphic addition and multiplication operations and cannot be used for complex computations in machine learning. With the advancement of machine learning, is there a method that allows training machine learning models directly on encrypted data? To the best of our knowledge, the current research literature in this particular area is virtually non-existent~\cite{DBLP:conf/mm/ZhangYY22}.

\section{Information Content of Images}\label{sec: Information Content of Images}
In the human eye, an image may contain information about a car, a group of cats, or a well-known celebrity. However, in the field of computer vision, the information content of an image is typically determined by the pixel values of each pixel point and the relationships between neighboring pixel points. The earliest multilayer perceptrons analyze image information based on the pixel values alone. Convolutional neural networks (CNNs) further consider the neighboring relationships between pixel points. In this section, we establish a theoretical framework for analyzing the variation in the information content of images.

\subsection{Definition}\label{subsec: Definition}
In order to accurately quantify the information content of an image, we establish the following requirements:

\begin{itemize}
\item The term ``information content of an image" does not have a strict definition. In this paper, we consider it to be composed of the pixel values and the neighboring relationships between pixel points.

\item The neighboring information of each pixel point depends on its distance from the surrounding pixel points, where closer distances result in greater neighboring information.

\item We focus on the variation in the information content of individual images (vertical comparison) rather than comparing the information content between different images (horizontal comparison).

\item Principle of verisimilitude: After various transformations, an image may make the objects within it more easily recognizable. However, we still regard the original image as having the highest information content because it is the most verisimilar. Regarding the rationale for this principle, we provide a detailed explanation in Section~\ref{Principle of Verisimilitude}.
\end{itemize}

We establish a Cartesian coordinate system, as shown in \figurename~\ref{Cartesian_coordinate}, where the image under study is denoted as $P$. $P$ has dimensions of $n$ columns and $m$ rows, with each pixel point at coordinates $(i,j)$ represented as $P_{ij}$, where $1\leq i\leq n, 1\leq j\leq m$. Let $M (\cdot)$ and $M_P (\cdot)$ denote the information content functions for images and pixel points, respectively, and,
\begin{eqnarray}\label{P and pij}
M(P)=\sum \limits_{i,j} M_P (P_{ij}).
\end{eqnarray}
\begin{figure}[!t]%
  \centering
  \includegraphics[width=0.4\textwidth]{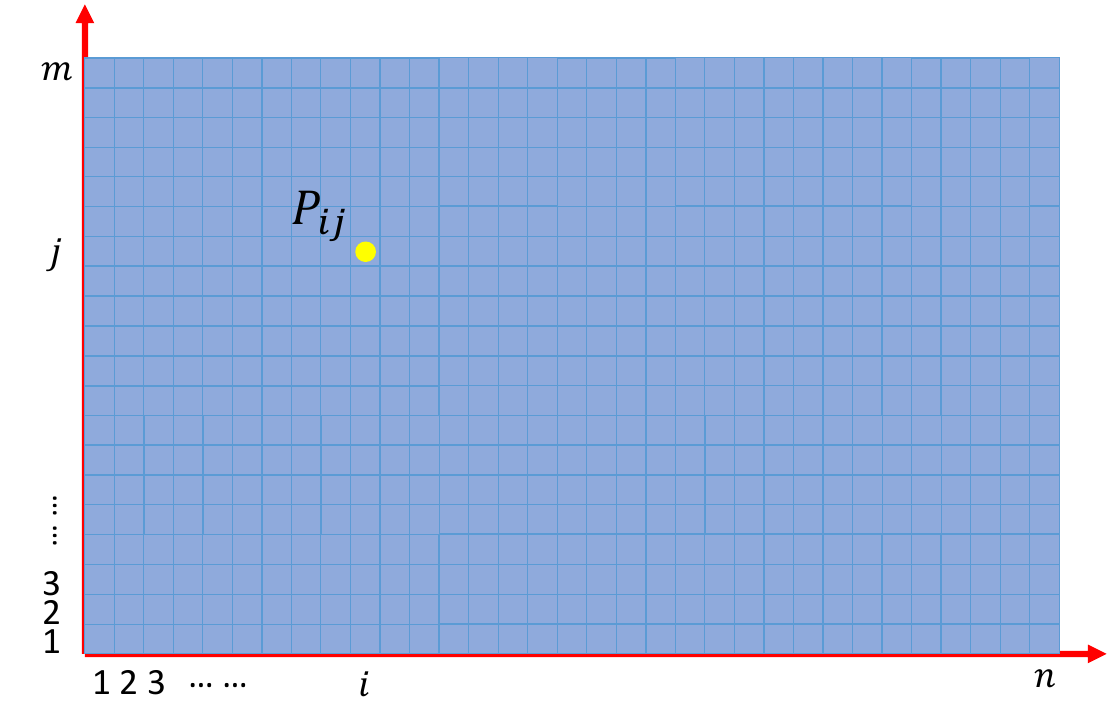}\\
  \caption{Cartesian coordinate system.}\label{Cartesian_coordinate}
\end{figure}

After that, the information content of each pixel point is composed of its pixel value and neighboring information. Let $M_{pix}(\cdot)$ and $M_{neig}(\cdot)$ respectively denote the functions for the pixel value information and neighboring information, i.e.,
\begin{eqnarray}\label{pij and pix neig}
M_P (P_{ij}) = M_{pix} (P_{ij}) + M_{neig} (P_{ij}).
\end{eqnarray}
In this research, we primarily focus on the neighboring information of pixel points, i.e., $M_{neig}(P_{ij})$. Next, we refine this function. Let $d(\cdot,\cdot)$, $\tilde{d}(\cdot,\cdot)$, and $d_{max}$ denote the distance between two pixel points, the modified distance, and the length of the diagonal, respectively. And,
\begin{eqnarray}\label{d d and dmax}
\left\{
\begin{aligned}
&d_{max} = \sqrt{m^2+n^2},\\
&d(P_{ij}, P_{st}) = \sqrt{(i-s)^2+(j-t)^2},\\
&\tilde{d}(P_{ij}, P_{st}) = \frac{d(P_{ij}, P_{st})}{d_{max}}.
\end{aligned}
\right.
\end{eqnarray}

Based on the distance Eq.~\eqref{d d and dmax}, we define $m_{neig}(\cdot,\cdot)$ as the function for measuring the information content between two pixel points. It is expressed as follows:
\begin{eqnarray}\label{m neig}
m_{neig}(P_{ij}, P_{st}) = 1 - \frac{1}{1+e^{6-18\tilde{d} (P_{ij}, P_{st})}}.
\end{eqnarray}

We have designed a Z-shaped information function by leveraging the rapid decay characteristics of the exponential function. This function represents the process of information decay between two pixels as the distance between them increases. The function $m_{neig}(\cdot,\cdot)$ is illustrated in Fig.~\ref{fig_m_neig}. Its intuitive interpretation is that for a given pixel point $P_{ij}$, the closer it is to another pixel point, the higher the information content between them, and this relationship is non-linear. Taking Fig.~\ref{lena_eyes} as an example, let's consider the pixel point at the center of the eye, denoted as $P_{ij}$. The information content between $P_{ij}$ and the pixel points in the surrounding small range (eye region) is significantly high, which is crucial for identifying the eye. On the other hand, the information content between $P_{ij}$ and the pixel points in the facial region is noticeably lower than in the eye region, and the information content between $P_{ij}$ and the pixel points outside the facial region is nearly $0$.

\begin{figure}[!t]%
  \centering
  \includegraphics[width=0.4\textwidth]{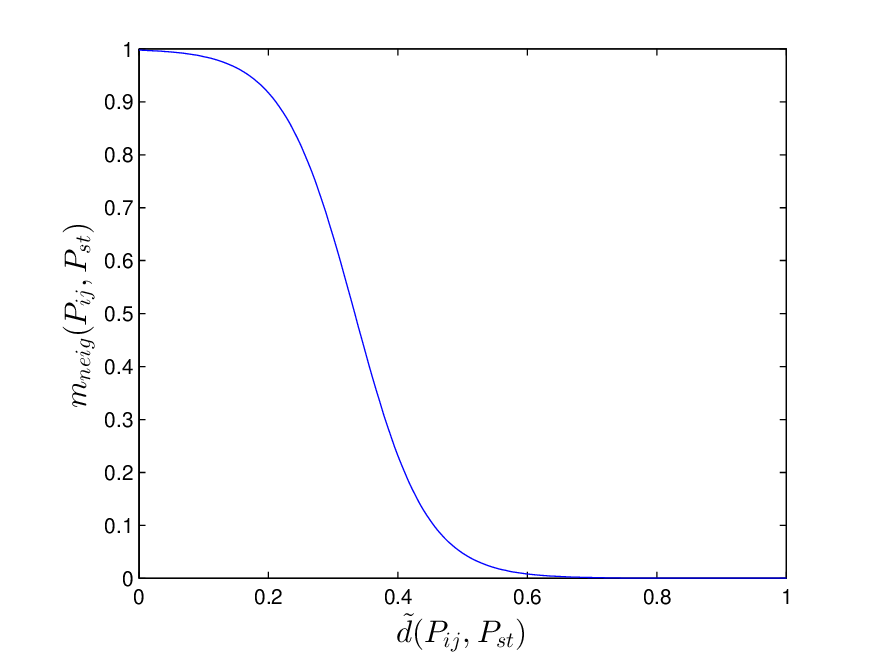}\\
  \caption{The graph of the function $m_{neig}(\cdot,\cdot)$.}\label{fig_m_neig}
\end{figure}

\begin{figure}[!t]%
  \centering
  \includegraphics[width=0.35\textwidth]{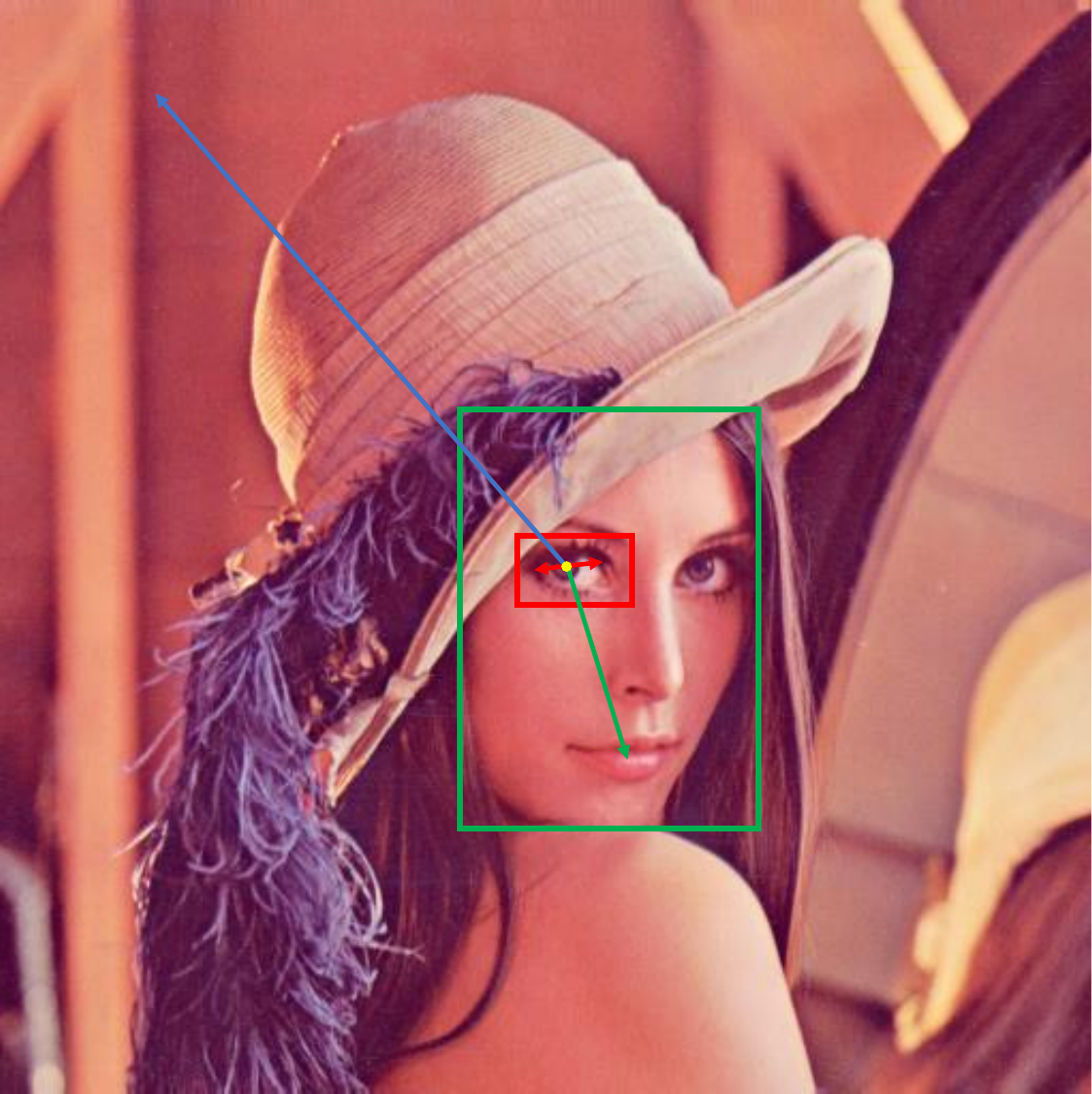}\\
  \caption{The neighboring information of each pixel is not solely determined by the surrounding eight pixels, but rather by the contribution of pixels at varying distances, which provide different amounts of neighboring information. As shown in the figure, in order to recognize the eyes, the red region provides the highest amount of information, while the green region utilizes facial information to provide a less amount of information. On the other hand, the background region (blue arrow) contributes almost negligible information.}\label{lena_eyes}
\end{figure}

Hence, the neighboring information $M_{neig}(P_{ij})$ of the pixel point $P_{ij}$ is equal to the sum of the information content between $P_{ij}$ and all pixel points in the image, i.e.,
\begin{eqnarray}\label{M neig 1}
M_{neig} (P_{ij}) = \sum \limits_{s,t} m_{neig}(P_{ij}, P_{st}).
\end{eqnarray}

Furthermore, when the pixel point positions in image $P$ undergo transformations, let $(i', j')$ represent the transformed position of $(i, j)$, and $m^*_{neig}(\cdot,\cdot)$ denote the information content function after the transformation, and
\begin{eqnarray}\label{m* neig}
&&m^*_{neig}(P_{i'j'}, P_{s't'}) \nonumber\\
&=& [1 - (m_{neig}(P_{ij}, P_{st}) - m_{neig}(P_{i'j'}, P_{s't'}))] \nonumber\\
&&\times m_{neig}(P_{ij}, P_{st}).
\end{eqnarray}
Specifically, when $(i, j) = (i', j')$, the equation $m^*_{neig}(P_{i'j'}, P_{s't'}) = m_{neig}(P_{ij}, P_{st})$ holds. Let the transformed image be denoted as $P^*$, and define the remaining information content function $\Upsilon(\cdot)$ as follows:
\begin{eqnarray}\label{m* neig}
\Upsilon(P^*) = \frac{M(P^*)}{M(P)}.
\end{eqnarray}
Since this framework focuses on studying the change in information content after image transformations, the function $\Upsilon(\cdot)$ can be employed to characterize the proportion of information content variation.

\subsection{Principle of Verisimilitude}\label{Principle of Verisimilitude}
To ensure the validity of the theory proposed in Section~\ref{subsec: Definition}, we also need to prove that the theory satisfies the ``principle of veracity".

\begin{theorem}\label{theorem}
Under the conditions stated in Section~\ref{subsec: Definition}, when swapping the positions of any two points in the image, the total information content of the image either decreases or remains unchanged.
\end{theorem}

\begin{proof}
Let $M (P)$ represent the information content of the image $P$. Consider any two pixel points in the image, denoted as $P_{i_1 j_1}$ and $P_{i_2 j_2}$. We will calculate the sum of the information content of these two points, denoted as:
\begin{align}\label{M neig i1j2+i2j2}
&M_{neig} (P_{i_1 j_1}) + M_{neig} (P_{i_2 j_2}) \\
&= \sum \limits_{s,t} m_{neig}(P_{i_1 j_1}, P_{st}) + \sum \limits_{s,t} m_{neig}(P_{i_2 j_2}, P_{st}) \nonumber\\
&= \sum \limits_{s,t} (m_{neig}(P_{i_1 j_1}, P_{st}) + m_{neig}(P_{i_2 j_2}, P_{st})) \nonumber\\
&= \sum \limits_{s,t} (1 - \frac{1}{1+e^{6-18\tilde{d} (P_{i_1 j_1}, P_{st})}} + 1 - \frac{1}{1+e^{6-18\tilde{d} (P_{i_2 j_2}, P_{st})}})\nonumber\\
&= \sum \limits_{s,t} (2 - \frac{1}{1+e^{6-18\tilde{d} (P_{i_1 j_1}, P_{st})}} - \frac{1}{1+e^{6-18\tilde{d} (P_{i_2 j_2}, P_{st})}}).\nonumber
\end{align}

Next, we will exchange the positions of $P_{i_1 j_1}$ and $P_{i_2 j_2}$. To facilitate distinction, let $P^*_{i j}$ represent the point $P_{i j}$ after the transformation. Specifically, the point $P_{i_1 j_1}$ with coordinates $(i_1, j_1)$ becomes $P^*_{i_1 j_1}$ with coordinates $(i_2, j_2)$. Consequently, the sum of the information content for $P^*_{i_1 j_1}$ and $P^*_{i_2 j_2}$ is as follows:
\begin{eqnarray}\label{M* neig i1j1+i2j2 1}
&&M_{neig} (P^*_{i_1 j_1}) + M_{neig} (P^*_{i_2 j_2}) \nonumber\\
&=& \sum \limits_{s,t} m^*_{neig}(P^*_{i_1 j_1}, P^*_{st}) + \sum \limits_{s,t} m^*_{neig}(P^*_{i_2 j_2}, P^*_{st}) \nonumber\\
&=& \sum \limits_{s,t} [1 - (m_{neig}(P_{i_1 j_1}, P_{st}) - m_{neig}(P_{i_2 j_2}, P_{st}))] \nonumber\\
&&\times m_{neig}(P_{i_1 j_1}, P_{st}) + \sum \limits_{s,t} [1 - (m_{neig}(P_{i_2 j_2}, P_{st}) \nonumber\\
&& - m_{neig}(P_{i_1 j_1}, P_{st}))] \times m_{neig}(P_{i_2 j_2}, P_{st}).
\end{eqnarray}
For the sake of simplicity in the derivation process, let's denote
\begin{eqnarray}\label{gamma 1 2}
\left\{
\begin{aligned}
&\gamma_1 = \frac{1}{1+e^{6-18\tilde{d} (P_{i_1 j_1}, P_{st})}},\\
&\gamma_2 = \frac{1}{1+e^{6-18\tilde{d} (P_{i_2 j_2}, P_{st})}},
\end{aligned}
\right.
\end{eqnarray}
then $0 < \gamma_1, \gamma_2 < 1$, we obtain
\begin{eqnarray}\label{M* neig i1j1+i2j2 2}
&&M_{neig} (P^*_{i_1 j_1}) + M_{neig} (P^*_{i_2 j_2}) \nonumber\\
&=& \sum \limits_{s,t} (1 + \gamma_1 - \gamma_2) \times (1 - \gamma_1) \nonumber\\
&&+ \sum \limits_{s,t} (1 + \gamma_2 - \gamma_1) \times (1 - \gamma_2) \nonumber\\
&=& \sum \limits_{s,t} [(2 - \gamma_1-\gamma_2) - (\gamma_1 - \gamma_2)^2] \nonumber\\
& \leq & \sum \limits_{s,t} [(2 - \gamma_1-\gamma_2) \nonumber\\
&=& M_{neig} (P_{i_1 j_1}) + M_{neig} (P_{i_2 j_2}).
\end{eqnarray}

When $(i,j) \neq {i_1 j_1}$ and $(i,j) \neq {i_2 j_2}$, we obtain
\begin{eqnarray}\label{M neig ij}
&&M_{neig} (P_{i j}) \nonumber\\
&=& \sum \limits_{s,t} m_{neig}(P_{i j}, P_{st}) \nonumber\\
&=& \sum \limits_{s,t} (1 - \frac{1}{1+e^{6-18\tilde{d} (P_{ij}, P_{st})}}), \nonumber\\
\end{eqnarray}
and
\begin{eqnarray}\label{M* neig ij 1}
&&M_{neig} (P^*_{i j}) \nonumber\\
&=& \sum \limits_{s,t} m^*_{neig}(P^*_{i j}, P^*_{s t}) \nonumber\\
&=& \sum \limits_{\substack{(s,t)\neq(i_1,j_1) \\ (s,t)\neq(i_2,j_2)}} m^*_{neig} (P^*_{i j}, P^*_{s t}) \nonumber\\
&& + m^*_{neig}(P^*_{i j}, P^*_{i_1 j_1}) + m^*_{neig}(P^*_{i j}, P^*_{i_2 j_2}) \nonumber\\
&=& \sum \limits_{\substack{(s,t)\neq(i_1,j_1) \\ (s,t)\neq(i_2,j_2)}} m_{neig} (P_{i j}, P_{s t}) \nonumber\\
&& + m^*_{neig}(P_{i j}, P^*_{i_1 j_1}) + m^*_{neig}(P_{i j}, P^*_{i_2 j_2}) \nonumber\\
&=& \sum \limits_{\substack{(s,t)\neq(i_1,j_1) \\ (s,t)\neq(i_2,j_2)}} m_{neig} (P_{i j}, P_{s t}) \nonumber\\
&& + [1 - (m_{neig}(P_{i j}, P_{i_1 j_1}) - m_{neig}(P_{i j}, P_{i_2 j_2}))] \nonumber\\
&& \times m_{neig}(P_{i j}, P_{i_1 j_1}) + m_{neig}(P_{i j}, P_{i_2 j_2}) \nonumber\\
&& \times [1 - (m_{neig}(P_{i j}, P_{i_2 j_2}) - m_{neig}(P_{i j}, P_{i_1 j_1}))].
\end{eqnarray}

For the sake of simplicity in the derivation process, let's denote
\begin{eqnarray}\label{gamma 3 4}
\left\{
\begin{aligned}
&\gamma_3 = \frac{1}{1+e^{6-18\tilde{d} (P_{ij}, P_{i_1 j_1})}},\\
&\gamma_4 = \frac{1}{1+e^{6-18\tilde{d} (P_{ij}, P_{i_2 j_2})}},
\end{aligned}
\right.
\end{eqnarray}
then
\begin{eqnarray}\label{M* neig ij 2}
&&M_{neig} (P^*_{i j}) \nonumber\\
&=& \sum \limits_{\substack{(s,t)\neq(i_1,j_1) \\ (s,t)\neq(i_2,j_2)}} m_{neig} (P_{i j}, P_{s t}) \nonumber\\
&& + [(1 - \gamma_3 + 1 -\gamma_4) - (\gamma_3 - \gamma_4)^2] \nonumber\\
&=& \sum \limits_{(s,t)} m_{neig} (P_{i j}, P_{s t}) - (\gamma_3 - \gamma_4)^2 \nonumber\\
&=& M_{neig} (P_{i j}) - (\gamma_3 - \gamma_4)^2 \nonumber\\
&\leq& M_{neig} (P_{i j}).
\end{eqnarray}

Let $P^*$ represent the image after swapping the two points. Referring to Eq.~\eqref{M neig i1j2+i2j2} to Eq.~\eqref{M* neig ij 2}, we obtain:
\begin{eqnarray}\label{MP* MP}
&& M(P^*) \nonumber\\
&=& \sum \limits_{i,j} M_P (P^*_{ij}) \nonumber\\
&=& \sum \limits_{i,j} (M_{pix} (P^*_{ij}) + M_{neig} (P^*_{ij})) \nonumber\\
&\leq& \sum \limits_{i,j} (M_{pix} (P_{ij}) + M_{neig} (P_{ij})) \nonumber\\
&=& \sum \limits_{i,j} M_P (P_{ij}) \nonumber\\
&=& M(P).
\end{eqnarray}

Therefore, the conclusion is valid.
\end{proof}

\subsection{Application Scenarios}
A comprehensive system for calculating image information content can be applied to various scenarios. Here are some application scenarios for reference:
\begin{itemize}
\item Evaluating and measuring image quality: Numerical calculations can provide a more rigorous assessment of image quality. Machine learning methods can be employed to learn the contrast, saturation, brightness, and sharpness of high-quality images, which can then be applied to adaptive image beautification.

\item Quantifying the degree of image distortion: Assessing the level of distortion in image enhancement, image reconstruction, and image restoration algorithms.

\item Evaluating the effectiveness of image compression and storage algorithms: By evaluating the redundancy and complexity of images, compression algorithms and storage strategies can be selectively chosen.

\item Assisting in feature extraction: As shown in \figurename~\ref{lena_eyes}, the perceived region sizes of different target objects vary from one region to another. By assessing the contribution of each region to the recognition of the target objects, the recognition regions for each target object can be determined, facilitating feature extraction.

\item Similarity measurement: Precise point-to-point similarity measurement is susceptible to minor disturbances, leading to significant errors. Similarity measurement methods based on spatial relationships can better assess the relationships between objects in images, particularly for measuring the similarity of rotated, flipped, or folded images.

\item Performance evaluation of image encryption algorithms: In image encryption, to protect against unauthorized recovery of the image, it is often necessary to evaluate the change in information content of encrypted images, thus assessing the performance of encryption algorithms.
\end{itemize}

In general, studying the information content of images enables the quantification and measurement of image content, quality, and features. This provides a foundation and guidance for algorithm design, performance evaluation, and data analysis in fields such as image processing, computer vision, and image encryption. In this paper, based on the proposed theory of image information content, we have designed an image encryption algorithm called Random Vortex Transform (RVT). RVT is an encryption scheme that relies on the neighborhood relationships of pixels. It fully utilizes the patterns of information content variation during the image transformation process, thereby preserving the invariance of image features throughout the encryption process. It differs from previous encryption paradigms because it allows direct feature extraction and target recognition using computer vision algorithms on the encrypted data without the need for decryption.

\section{Random Vortex Transformation}\label{sec: Random Vortex Transformation}
Theorem~\ref{theorem} not only demonstrates that swapping the positions of pixels decreases the quality of the image (i.e., reduces the information content), but it also unveils an underlying law: the reduction in information content is positively correlated with the distance between the two pixels (as deduced from the ``$\leq$'' in Eq.~\eqref{M* neig i1j1+i2j2 2}, Eq.~\eqref{M* neig ij 2} and Eq.~\eqref{MP* MP}), that is $(\tilde{d} (P_{i_1 j_1}, P_{st}) - \tilde{d} (P_{i_2 j_2}, P_{st}))$. This implies that if the change in distance between these two points (or multiple points) is smaller than a certain threshold, the information content of the image remains nearly unchanged after the position swap. Consequently, machine learning methods can still be used to recognize the target information in the image. On the other hand, swapping the positions of certain pixel points in the image can render the target information unrecognizable to the human eye. Figure~\ref{Objectives_of_encryption_schemes} illustrates the design objectives of this encryption scheme.
\begin{figure}[!t]%
  \centering
  \includegraphics[width=0.5\textwidth]{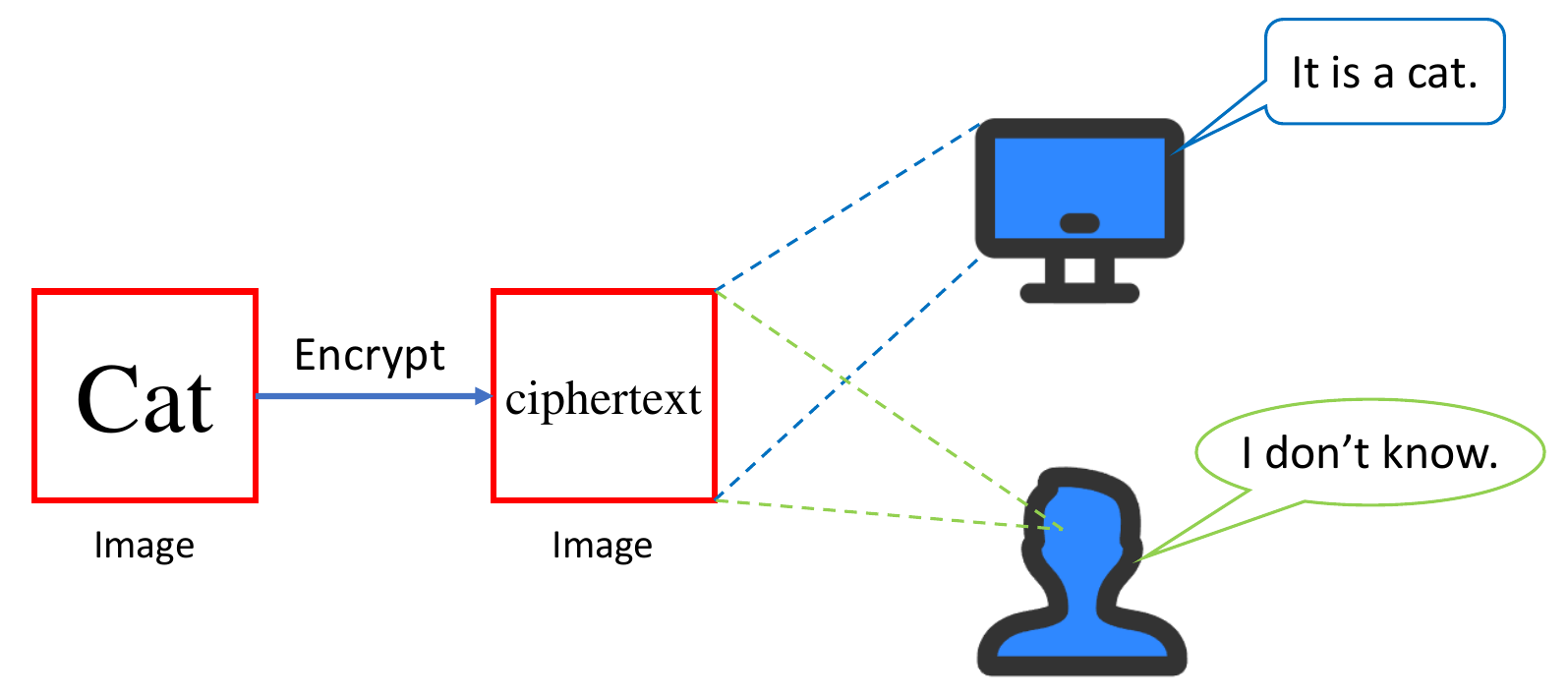}\\
  \caption{Design objectives of the encryption scheme.}\label{Objectives_of_encryption_schemes}
\end{figure}

Based on the aforementioned findings, we propose a technique for encryption called the Random Vortex Transformation. Our approach is to utilize random vortex transformations, where we permute the coordinates of selected points in the image. This permutation makes it difficult for the human eye to identify the target information in the image, while still allowing accurate recognition by a computer. Utilizing this principle to achieve image encryption. The mathematical expression for the Random Vortex Transformation is as follows:
\begin{eqnarray}\label{Vortex}
\left\{
\begin{aligned}
&i' = d \cdot cos (\theta+ (R - d)\cdot f(d)),\\
&j' = d \cdot sin (\theta+ (R - d)\cdot f(d)),
\end{aligned}
\right.
\end{eqnarray}
where $(i', j')$ represents the coordinates of point $P_{ij}$ after undergoing a random vortex transformation, $d = d(P_{ij}, P_{i_0j_0})$, where $P_{i_0j_0}$ denotes the vortex center, and $\theta = \arctan\left(\frac{j - j_0}{i - i_0}\right)$. $R$ denotes the vortex radius, and $d\leq R\leq \min(n - i, i, m - j, j)$. $f(\cdot)$ is a random function with a bounded derivative, i.e. $|f'(\cdot)|<\omega$, also referred to as the vortex coefficient function, where $\omega$ is a constant. Various methods can be employed to generate the random function, such as a polynomial function with random coefficients or a sum of cosine functions with randomly selected parameter values.

Figure~\ref{lena_vortex} illustrates the result of applying random vortex transformations to the Lena image. It is evident that the transformed Lena image is no longer recognizable to the human eye, thereby achieving the objective of encryption and safeguarding the right to privacy of facial portraits. Figure~\ref{MNIST-vorte} demonstrates the vortex effects on a subset of images from the MNIST dataset. Upon visual observation, it can be discerned that although the images after the vortex transformation are distinct from their original counterparts, there still exists a commonality among images with the same class label. This commonality serves as one of the crucial foundations for computer recognition of the target objects.
\begin{figure}[!t]%
  \centering
  \includegraphics[width=0.3\textwidth]{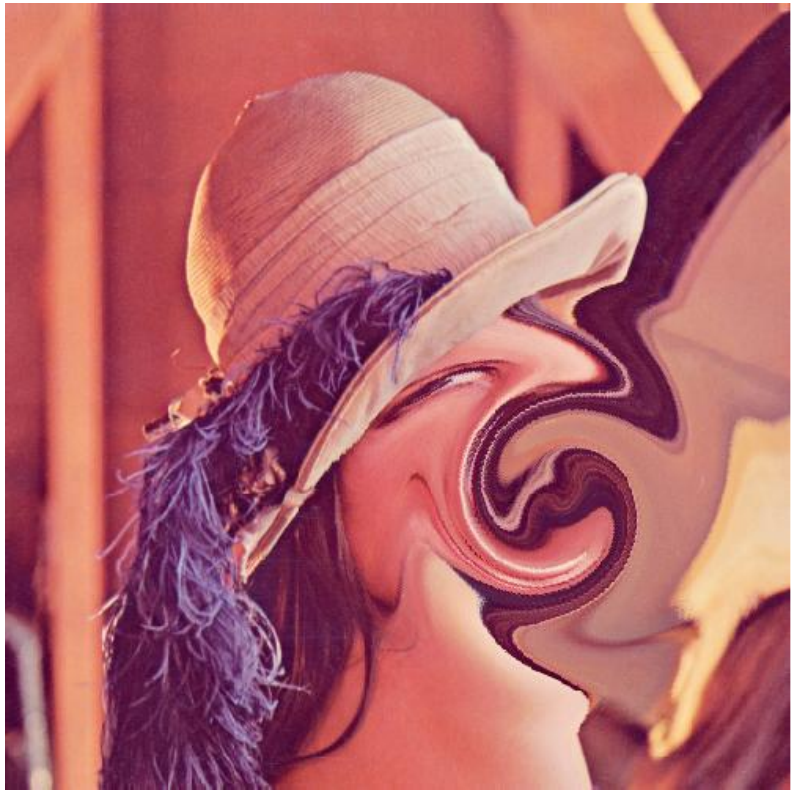}\\
  \caption{Lena after vortex transformation.}\label{lena_vortex}
\end{figure}

\begin{figure}[!t]%
  \centering
  \includegraphics[width=0.4\textwidth]{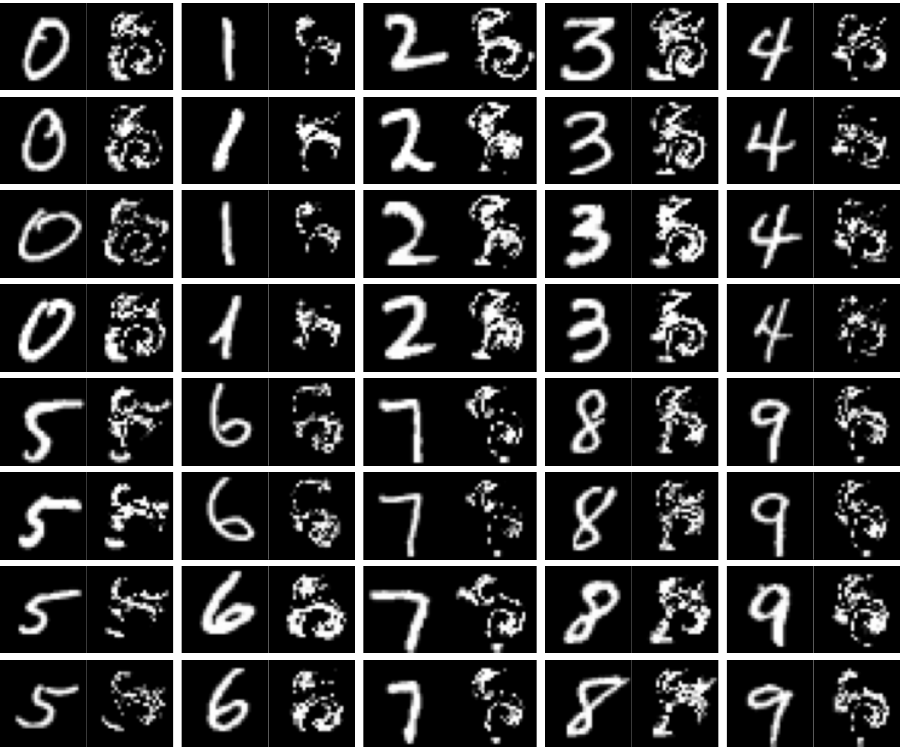}\\
  \caption{In each subplot, the left image represents the original image from the MNIST dataset, while the right image portrays the same image after undergoing vortex transformation.}\label{MNIST-vorte}
\end{figure}

Subsequently, we will elucidate Eq.~\eqref{Vortex} to aid readers in comprehending the vortex transformation.

\begin{itemize}
\item We designate the points ``at an equidistant distance from the vortex center" as ``points on the same circumference". The purpose of vortex transformation is to rotate pixels along the same circumference, thereby distorting the target information in the image. Figure~\ref{Example of permuting pixel positions} demonstrates an example of this process.

\item The random function $f(\cdot)$ plays a crucial role in image encryption as it is used to assign distinct vortex coefficients to each circumference. As a result, the vortex coefficients for each circumference are randomized, significantly enhancing the difficulty of image recovery.

\item We constrain the derivative of the random function $f(\cdot)$ to be bounded to prevent the destruction of neighboring information on different circumferences due to excessively large vortex coefficients. In other words, we aim to avoid a significant reduction in the amount of information in the image caused by the magnitudes of the term $(\tilde{d} (P_{i_1 j_1}, P_{st}) - \tilde{d} (P_{i_2 j_2}, P_{st}))$ in Theorem~\ref{theorem} becoming too large due to the vortex transformation. It is precisely based on this configuration that, after undergoing random vortex transformations, the neighboring information of the points is not excessively disrupted, enabling the computer to still recognize the target objects based on the pixel values and proximity information of the pixels. Figure~\ref{A graph of a randomly generated function} illustrates a graph generated from a randomly generated function.

\item The term $(R - d)$ is introduced to prevent excessively large vortex coefficients at the edges of the vortex.
\end{itemize}
\begin{figure}[!t]%
  \centering
  \includegraphics[width=0.3\textwidth]{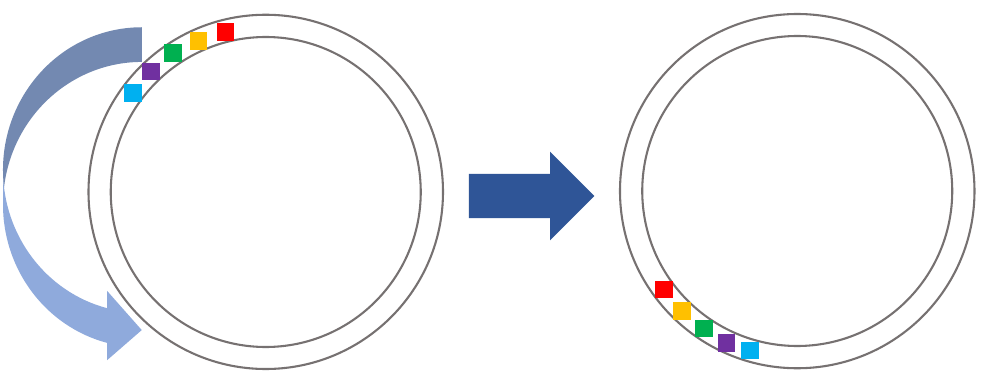}\\
  \caption{Example of permuting pixel positions.}\label{Example of permuting pixel positions}
\end{figure}
\begin{figure}[!t]%
  \centering
  \includegraphics[width=0.4\textwidth]{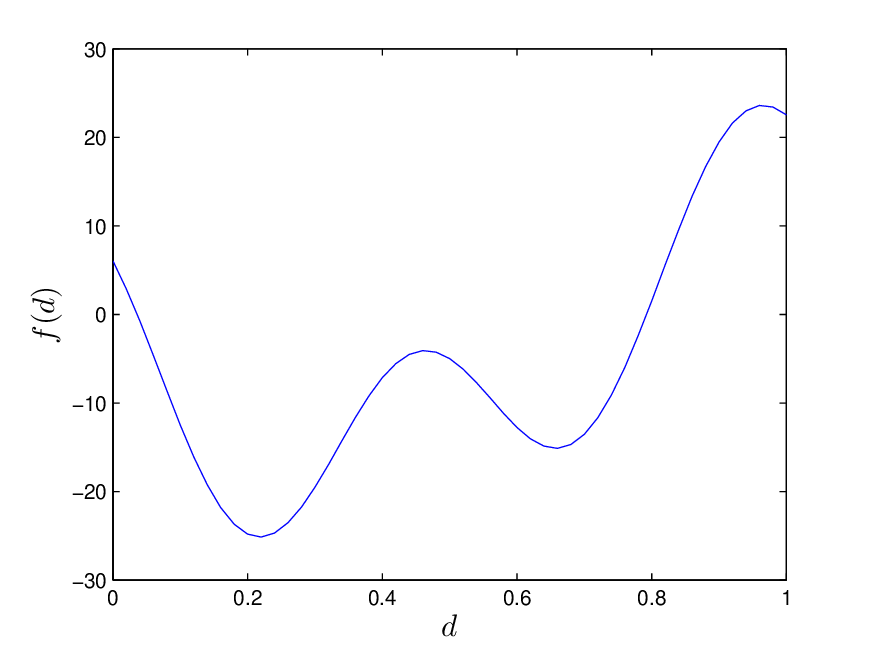}\\
  \caption{A graph of a randomly generated function.}\label{A graph of a randomly generated function}
\end{figure}

Due to the low cost (only resulting in minimal loss of pixel neighborhood information) of random vortex transformations for image encryption, it is possible to apply multiple random vortex transformations on the same image to achieve enhanced encryption. Equation \eqref{Vortex} is abbreviated as follows:
\begin{eqnarray}\label{simple Vortex}
P^* &=& G(P, [(i_0, j_0), R, f(\cdot)])\nonumber\\
&=& G(P,~\cdot~),
\end{eqnarray}
where, $G(P,~\cdot~)$ represents the abstract expression of Eq.~\eqref{Vortex}, and for each vortex transformation, the parameters $[(i_0, j_0), R, f(\cdot)]$ are random. The superposed random vortex transformation can be expressed as:
\begin{eqnarray}\label{superposition Vortex}
P^* = G_{Sup}(~\cdots~ G_2(G_1(P,~\cdot~),~\cdot~)~\cdots~),
\end{eqnarray}
where, $Sup$ denotes the number of superposed random vortex transformations.

\section{Experiments}\label{sec: Experiments}
In section~\ref{sec: Information Content of Images}, the feasibility of the random vortex transformation is theoretically analyzed. In this section, we aim to empirically validate this theory using the machine learning model. We indirectly examine the neighboring information of images by evaluating the model's recognition accuracy on the images. Our experimental design is illustrated in \figurename~\ref{Experimental design}. Specifically, we encrypt the images using the proposed random vortex transformation described in this paper. Subsequently, we directly train a neural network model using the encrypted images and evaluate its accuracy on the test dataset. On the other hand, we train the same neural network using the original data. By comparing the results from both sets of experiments, we can validate the conclusion of whether the ``random vortex transformation disrupts the neighboring information of images".

\begin{figure}[!t]%
  \centering
  \includegraphics[width=0.5\textwidth]{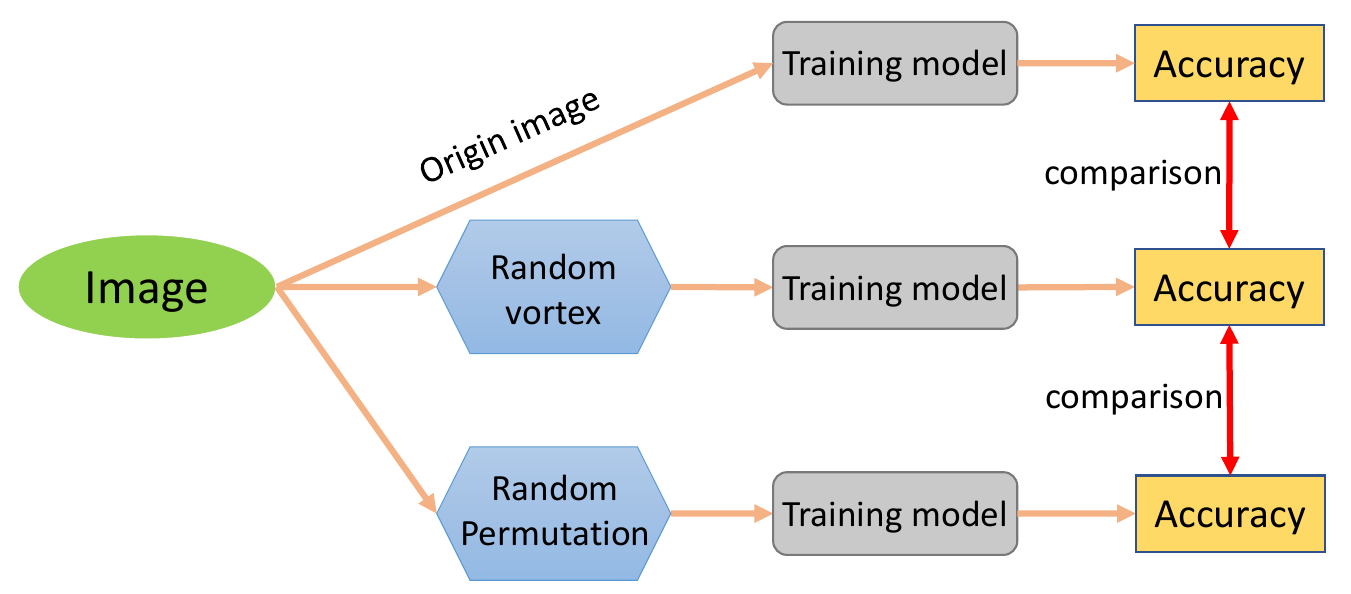}\\
  \caption{The experimental design strategy involves indirectly representing the proportion of preserved information during the encryption process through the accuracy of the model.}\label{Experimental design}
\end{figure}

To examine the impact of different random vortex transformations on the experimental results, we generate multiple sets of random vortices for encrypting the images. We then compare the fluctuations in the final test accuracy. This section is referred to as the sensitivity experiment of random function.

All the models are trained with a mini-batch of 128 on one GPU RTX 3090.

\subsection{Experimental Settings}

In this section, we will introduce the various components depicted in \figurename~\ref{Experimental design}.

\subsubsection{Datasets}

In our experiment, three datasets are used to verify the effectiveness of the vortex transformation, namely MNIST, Fashion and CIFAR-10.

\begin{itemize}
\item The MNIST dataset ~\cite{DBLP:journals/pieee/LeCunBBH98} is a classical dataset for image classification, which contains $10$ categories of $28 \times 28$ handwritten digit images from $1$ to $10$. It has $70000$ samples in total, where $60000$ samples for the training set, and $10000$ samples for the test set.

\item The Fashion dataset~\cite{DBLP:journals/corr/abs-1708-07747} is a MNIST-like dataset with images of fashion objects, which also has $10$ categories, like T-shirt, trouser etc. The training set and the test set are composed of $60000$ and $10000$ samples respectively.

\item The CIFAR-10 dataset ~\cite{krizhevsky2009learning} is a more complex dataset consisting of $60000$ $32 \times 32$ images for $10$ categories, with $6000$ images per category. The training set contains $50000$ images and the test set contains $10000$ images.
\end{itemize}

\subsubsection{Model}

In our experiment, we adopt two powerful neural networks ResNet-18 and Vision Transformer to do image classification.

\begin{itemize}
  \item ResNet-18~\cite{DBLP:conf/cvpr/HeZRS16} is an architecture introducing short-cut connection between layers which enables the network to be built deeper. Based on the architecture of ResNet-18 for ImageNet implemented in~\cite{DBLP:conf/cvpr/HeZRS16}, concerning the small size of the images in the datasets we used, i.e. $28 \times 28$ in MNIST and Fahion and $32 \times 32$ in CIFAR-10, much smaller than the images in CIFAR-10, we do some small adjustments to the original architecture, replacing the first $7 \times 7$ convolution with a $3 \times 3$ convolution and remove the followed maxpooling layer, which means that the size of the feature map will not be changed after the first convolution layer.

  \item Vision Transformer~\cite{DBLP:conf/iclr/DosovitskiyB0WZ21} (VIT) is a transformer-based architecture for image recognition that splits the input images into patches and trains them like a sequence of `words'. In our experiment, we adopt the ViT-Base model in~\cite{DBLP:conf/iclr/DosovitskiyB0WZ21}, with $12$ attention layers.
\end{itemize}

\subsubsection{Compared Methods}

For each dataset, we train models using the original data, the data encrypted using vortex transformations, and the data with randomly permuted pixels. We compare the classification accuracy of these three sets of experiments. The details are described as follows:

\begin{itemize}
\item Origin: We directly train models on three datasets, using them as the baseline for information content of each dataset.

\item Random Vortex Transformation (abbreviated as Vortex): We perform multiple sets of random vortex transformations on each dataset. For example, in the sensitivity experiment of random functions, we conduct three sets of experiments on the MNIST dataset. Each set of experiments involves $4-5$ random vortex transformations. The vortex center, vortex radius, and random function are defined as Eq.~\eqref{eq random function 1} - \eqref{eq random function 3}. It is important to note that the strength of the random function lies in its arbitrary number of parameters, and only a partial list of random functions is provided here.
\begin{figure*}
  \centering
\begin{small}
      \begin{equation}\label{eq random function 1}
      \left\{
      \begin{aligned}
      &(i', j') = (5, 11), R = 4, f(d) = -1.88\sin(1.20d+0.95) +1.17\cos(0.68d+0.74),\\
      &(i', j') = (18, 18), R = 9, f(d) = 1.96d + 0.03d^3 + 1.67e^{2d} +1.79cos(1.26d+1.60), \\
      &(i', j') = (10, 6), R = 5, f(d) = 0.09\sqrt{d} + 0.16d^2 + 1.89e^d - 0.52sin(1.46d+1.15), \\
      &(i', j') = (8, 12), R = 7, f(d)= 1.16d^5 + 0.03\ln(d+1) +1.40 e^d + 0.60cos(0.67d+0.92), \\
      &(i', j') = (19, 15), R = 8, f(d) = 1.64\cdot2^d + 3.48d +0.01\lg(d+1) + 1.65sin(1.93d+0.88).\\
      \end{aligned}
      \right.
      \end{equation}
    \end{small}
\end{figure*}
\begin{figure*}
  \centering
\begin{small}
      \begin{equation}\label{eq random function 2}
      \left\{
      \begin{aligned}
      &(i', j') = (12, 26), R = 1, f(d) = 0.50d + 1.49d^3 + 0.99e^{2d} + 0.60cos(0.22d+0.08),\\
      &(i', j') = (19, 21), R = 6, f(d) = 1.70\sqrt{d} + 1.49d^2 + 0.32e^d - 1.50sin(0.37d+1.51), \\
      &(i', j') = (8, 10), R = 7, f(d)=  1.61d^5 + 1.50\ln(d+1) + 0.77e^d + 1.05cos(1.30d-0.34), \\
      &(i', j') = (12, 11), R = 10, f(d) = 0.17\cdot2^d + 2.86d + 0.49\lg(d+1) - 0.49sin(1.24d+0.86). \\
      \end{aligned}
      \right.
      \end{equation}
    \end{small}
\end{figure*}
\begin{figure*}
  \centering
\begin{small}
      \begin{equation}\label{eq random function 3}
        \left\{
        \begin{aligned}
        &(i', j') = (23, 22), R = 4, f(d) = -1.56\sin(0.58d+ 0.81) + 0.83\cos(0.06d+0.84),\\
        &(i', j') = (8, 12), R = 7, f(d) = 1.79d+ 1.71d^3 + 0.66e^{2d} + 0.72cos(1.26d+1.67),\\
        &(i', j') = (14, 10), R = 9, f(d) = 1.20\sqrt{d} + 0.49d^2 +1.41e^d + 0.72sin(0.23d+1.58), \\
        &(i', j') = (15, 26), R = 1, f(d)= 1.74d^5 + 1.87\ln(d+1) + 1.84 e^d + 0.75cos(1.17d-0.72), \\
        &(i', j') = (20, 19), R = 7, f(d) = 0.04\cdot2^d + 0.86d + 0.13\lg(d+1) + 0.52sin(1.49d+0.56). \\
        \end{aligned}
        \right.
      \end{equation}
    \end{small}
\end{figure*}

\item Random Permutation (abbreviated as Random): As comparative experiment, we randomly permute the pixels within each image. This approach preserves the pixel value information while disrupting the neighboring information between pixels. By contrasting random permutation with random vortex transformations, we can test the ability of random vortex transformations to preserve neighboring information.
\end{itemize}

%
%
%

\subsection{Results and Analysis}
\subsubsection{Comparison of Results after Vortex Encryption}
\begin{figure*}[!t]%
  \centering
  \includegraphics[width=0.9\textwidth]{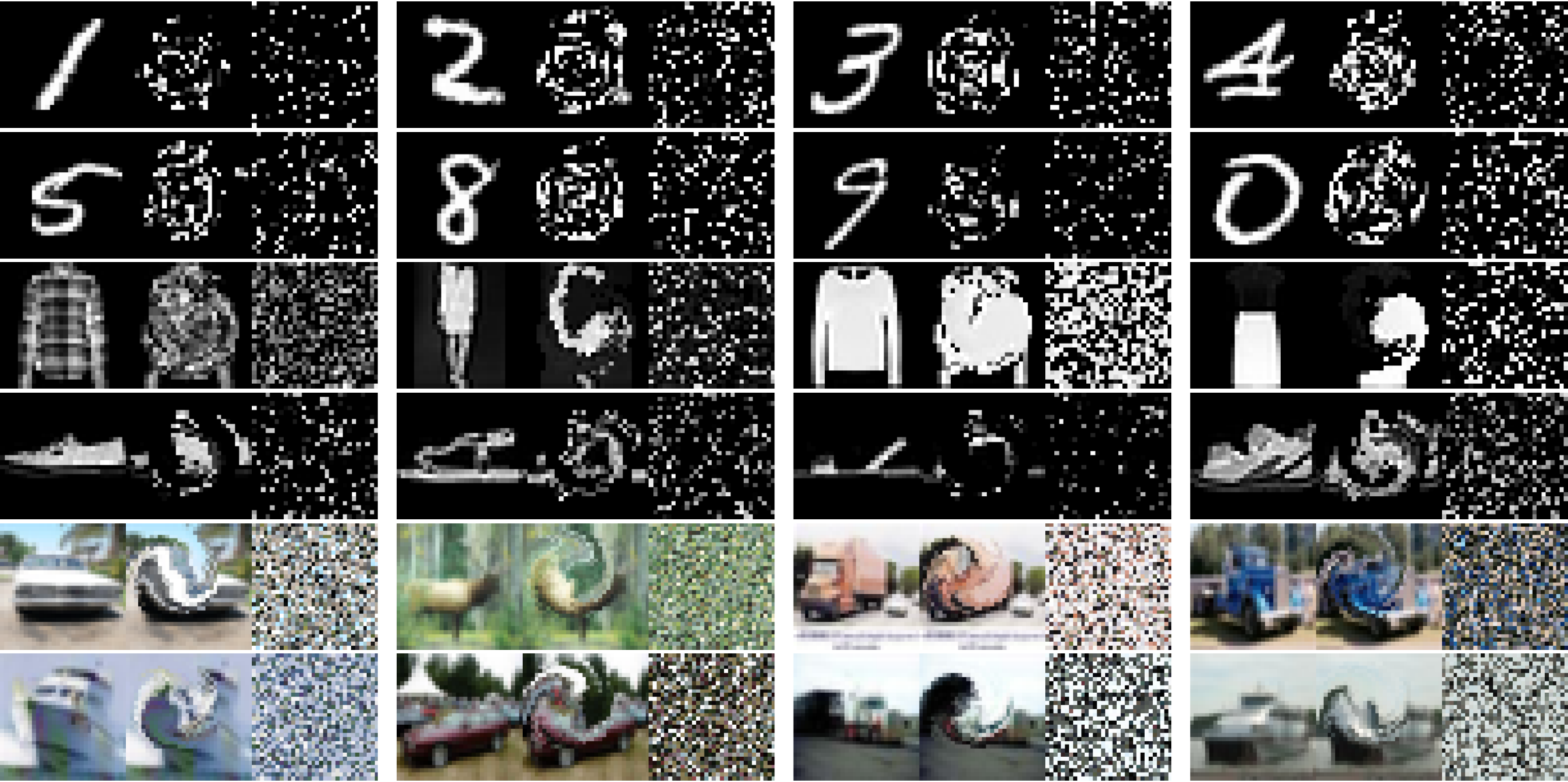}\\
  \caption{The image presentation of three datasets is shown. The first and second rows are from MNIST, the third and fourth rows are from Fashion, and the fifth and sixth rows are from CIFAR-10. Each subplot consists of three images: the original image, the image after vortex transformation, and the image after random pixel permutation.}\label{Three_dataset_image_presentation}
\end{figure*}

\begin{table*}[!t]
  \begin{centering}
    \caption{Comparison result of Origin, Vortex and Random.}\label{tab three dataset}
    \begin{tabular}{@{}c@{}|@{}c@{}c@{}c@{}|@{}c@{}c@{}c@{}|@{}c@{}c@{}c@{}}
      \hline
      \multirow{4}{*}{Model}     & \multicolumn{3}{c|}{\multirow{2}{*}{MNIST}}                                                                                              & \multicolumn{3}{c|}{\multirow{2}{*}{Fashion}}                                                                                 & \multicolumn{3}{c}{\multirow{2}{*}{CIFAR-10}}                                                                                           \\
                                & \multicolumn{3}{c|}{}                                                                                                                    & \multicolumn{3}{c|}{}                                                                                                         & \multicolumn{3}{c}{}                                                                                                                    \\ \cline{2-10}
                                & \multicolumn{1}{c}{\multirow{2}{*}{Origin}} & \multicolumn{1}{c}{\multirow{2}{*}{Vortex}} & \multicolumn{1}{c|}{\multirow{2}{*}{Random}} & \multicolumn{1}{c}{\multirow{2}{*}{Origin}} & \multicolumn{1}{c}{\multirow{2}{*}{Vortex}} & \multirow{2}{*}{Random}           & \multicolumn{1}{c}{\multirow{2}{*}{Origin}} & \multicolumn{1}{c}{\multirow{2}{*}{Vortex}} & \multicolumn{1}{c}{\multirow{2}{*}{Random}} \\
                                & \multicolumn{1}{c}{}                        & \multicolumn{1}{c}{}                        & \multicolumn{1}{c|}{}                        & \multicolumn{1}{c}{}                        & \multicolumn{1}{c}{}                        &                                   & \multicolumn{1}{c}{}                        & \multicolumn{1}{c}{}                        & \multicolumn{1}{c}{}                        \\ \hline
      \multirow{2}{*}{ResNet-18} & \multirow{2}{*}{99.20\%}                    & \multirow{2}{*}{$98.89\%^{\green{\downarrow 0.31\%}}$}            & \multirow{2}{*}{$28.99\%^{\green{\downarrow 70.21\%}}$}            & \multirow{2}{*}{94.12\%}                    & \multirow{2}{*}{$92.85\%^{\green{\downarrow 1.27\%}}$}   & \multirow{2}{*}{$40.10\%^{\green{\downarrow 54.02\%}}$}    & \multirow{2}{*}{92.80\%}                    & \multirow{2}{*}{$87.56\%^{\green{\downarrow 5.24\%}}$}   & \multirow{2}{*}{$42.80\%^{\green{\downarrow 50.0\%}}$}             \\
                                &                                             &                                             &                                              &                                             &                                             &                                   &                                             &                                             &                                             \\
      \multirow{2}{*}{ViT}       & \multirow{2}{*}{99.47\%}                    & \multirow{2}{*}{$99.15\%^{\green{\downarrow 0.32\%}}$}            & \multirow{2}{*}{$31.84\%^{\green{\downarrow 67.63\%}}$}              & \multirow{2}{*}{94.95\%}                    & \multirow{2}{*}{$92.61\%^{\green{\downarrow 2.34\%}}$}     & \multirow{2}{*}{$56.19\%^{\green{\downarrow 38.76\%}}$} & \multirow{2}{*}{98.50\%}                    & \multirow{2}{*}{$92.04\%^{\green{\downarrow 6.46\%}}$}   & \multirow{2}{*}{$50.07\%^{\green{\downarrow 48.43\%}}$}             \\
                                &                                             &                                             &                                              &                                             &                                             &                                   &                                             &                                             &                                             \\ \hline
 $\Upsilon(\cdot)$  & ~100.00\%~                                    & ~{$95.25\%^{\green{\downarrow 4.75\%}}$}~                                      & ~{$67.34\%^{\green{\downarrow 32.66\%}}$}~                                       & ~100.00\%~                                    & ~{$90.04\%^{\green{\downarrow 9.96\%}}$}~                                     & \multicolumn{1}{l|}{$67.23\%^{\green{\downarrow 32.77\%}}$}   &~ 100.00\%~                                    & ~{$97.38\%^{\green{\downarrow 2.62\%}}$}~                                     & ~{$67.08\%^{\green{\downarrow 32.92\%}}$}~                                     \\ \hline
    \end{tabular}
    \end{centering}
\end{table*}

Figure~\ref{Three_dataset_image_presentation} showcases the visual effects of images from three datasets after undergoing random vortex transformations. Corresponding test accuracies are presented in Table~\ref{tab three dataset}. From the \figurename~\ref{Three_dataset_image_presentation} and Table~\ref{tab three dataset}, we can draw the following conclusions:

\begin{itemize}
    \item The accuracy of MNIST-Vortex and Fashion-Vortex is only $0\%$ and $2.5\%$ lower compared with that of MNIST-Origin and Fashion-Origin. This indicates that random vortex transformations can render images indistinguishable to the human eye while maintaining recognizability by a computer.

    \item CIFAR-10-Vortex exhibits a decrease of approximately $6\%$ compared to CIFAR-10-Origin. This is attributed to CIFAR-10 being a three-channel dataset, resulting in a threefold increase in the disruption of neighboring pixel information. Nevertheless, CIFAR-10-Vortex maintains a high accuracy ranging from $87.56\%$ to $92.04\%$.

    \item Under random pixel permutation, the accuracy is significantly low. This is because random permutation completely disrupts the neighboring relationships between pixels, making the images difficult to recognize. On the other hand, this also signifies that the advantageous experimental results achieved by random vortex transformations are not solely attributed to the model's strong performance, but also due to the ability of random vortex transformations to preserve a significant amount of neighboring information in the images.
\end{itemize}

Figure~\ref{Testing_error_curve} illustrates the test error variation curves for six experimental groups. The datasets after vortex transformations exhibit slower convergence; however, they still manage to converge and achieve satisfactory outcomes.

\begin{figure*}[!t]%
  \centering
  \includegraphics[width=1.0\textwidth]{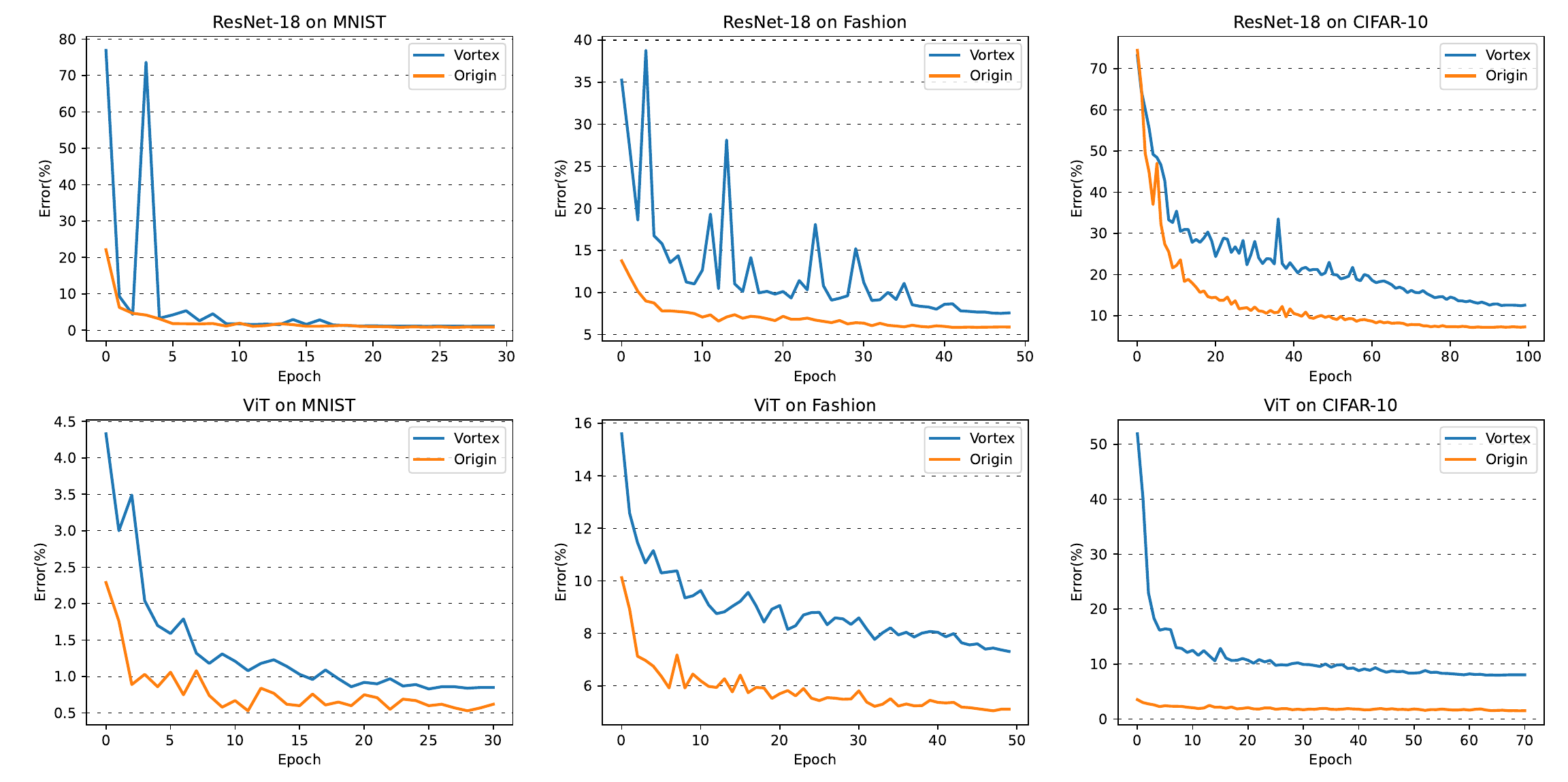}\\
  \caption{The curve of testing error variation.}\label{Testing_error_curve}
\end{figure*}

\subsubsection{Variation in Information Content}
Since this paper primarily investigates the neighboring information of pixels, and the image transformations employed in the experiments do not alter the pixel values, the information content discussed here disregards the pixel values. In other words, Eq.~\eqref{pij and pix neig}, which involves the function $M_{pix} (P_{ij})$, is not taken into consideration in this context.

The last row of Table~\ref{tab three dataset} displays the remaining information content $\Upsilon(\cdot)$ after the image undergoes transformations. Comparing random vortex transformations and random permutations, the decrease in accuracy is roughly proportional to the loss rate of information content. This suggests that the random vortex transformations, designed based on information content, are reasonable. However, it is important to note that accuracy is only an indirect reflection of information content, and thus, the relationship between the two is not strictly linear.

Figure~\ref{fig_Information_content_on_dataset} illustrates the variation in information content after multiple random vortex transformations and random permutations for three datasets. After random permutations, the relative positions of pixels are severely disrupted, leading to a decay of information content in the image to approximately 68\%. Conversely, due to the ability of random vortex transformations to preserve the neighboring relationships within the image, the decay rate of information content is significantly slower.

\begin{figure*}[!t]%
  \centering
  \includegraphics[width=1.0\textwidth]{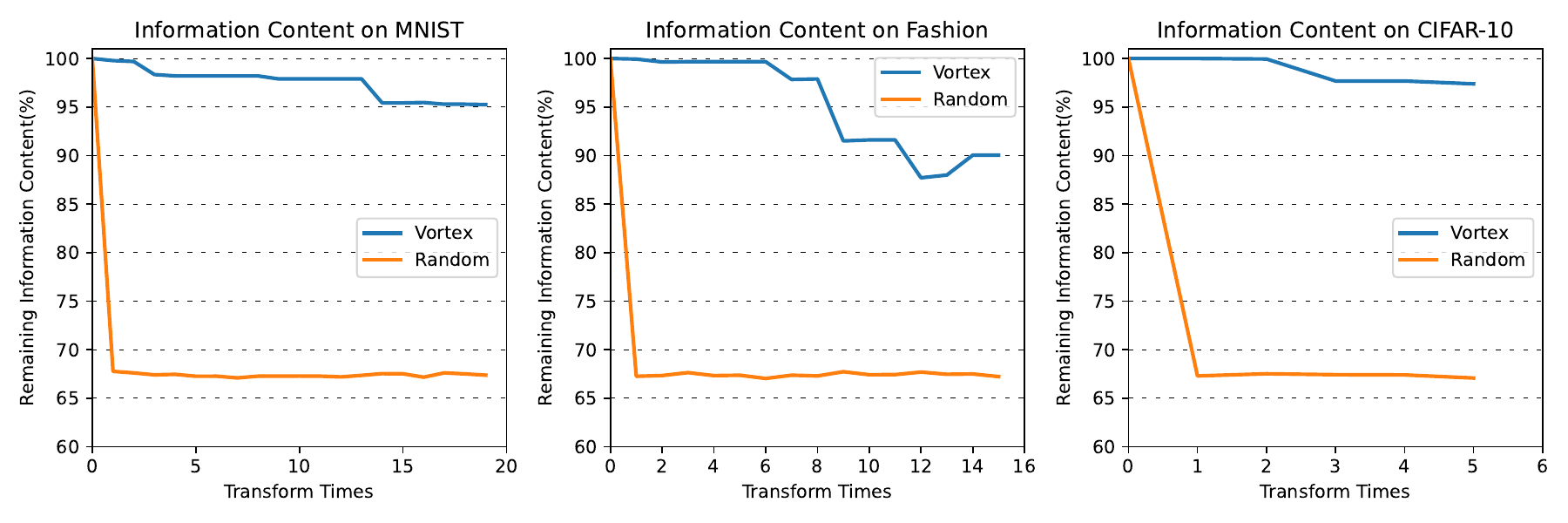}\\
  \caption{The curve represents the variation in information content after the image undergoes transformations. The horizontal axis represents the number of transformations (random vortex transformations or random permutations), while the vertical axis represents the proportion of information content after transformation compared to the original image.}\label{fig_Information_content_on_dataset}
\end{figure*}

\subsubsection{Sensitive Experiment}
In Table~\ref{tab Sensitive Experiment mnist}, we compare the experimental results of three sets with different random parameters to verify the stability of the random vortex transformations. From the experimental results, it can be observed that despite using different random functions and other parameters for encryption in each experiment, they all achieve accuracy close to the ``Origin". This indicates that the encryption scheme can steadily maintain the neighboring relationships between pixels.
\begin{table}[]
\centering
  \caption{Comparison result of three versions of vortex transformation on MNIST.}\label{tab Sensitive Experiment mnist}
  \begin{tabular}{c|llll}
  \hline
  \multirow{2}{*}{Model}     &\multicolumn{1}{c}{\multirow{2}{*}{Origin}} & \multicolumn{1}{c}{\multirow{2}{*}{Vortex-1}} & \multicolumn{1}{c}{\multirow{2}{*}{Vortex-2}} & \multicolumn{1}{c}{\multirow{2}{*}{Vortex-3}} \\
                             & \multicolumn{1}{c}{}         & \multicolumn{1}{c}{}                         & \multicolumn{1}{c}{}                         & \multicolumn{1}{c}{}                         \\ \hline
  \multirow{2}{*}{ResNet-18} & \multirow{2}{*}{99.20\%} & \multirow{2}{*}{98.77\%}             & \multirow{2}{*}{98.67\%}             & \multirow{2}{*}{98.75\%}             \\
                             & &                                              &                                              &                                              \\
  \multirow{2}{*}{ViT}       & \multirow{2}{*}{99.47\%} & \multirow{2}{*}{98.76\%}             & \multirow{2}{*}{98.49\%}             & \multirow{2}{*}{98.62\%}             \\
                             & &                                              &                                              &                                              \\ \hline
  \end{tabular}
  \end{table}



\section{Conclusion}\label{sec: Conclusion}
In this paper, we propose a theoretical framework for measuring the information content of images, inspired by convolutional neural networks and Vision Transformers. This framework considers the entire set of pixels in an image as a whole, taking into account the relative positional relationships among the pixels. Instead of considering the absolute position of individual pixels or the correlation coefficients between each pixel and its surrounding eight pixels in isolation. This measurement approach plays a significant role in various areas such as image quality assessment, image feature extraction, and evaluation of image encryption algorithms. Furthermore, we explore image encryption algorithms based on this framework. Common image encryption schemes aim to disrupt the neighboring relationships between pixels as much as possible. However, we take a completely opposite approach by designing a scheme that encrypts images while preserving the neighboring information as much as possible --- the Random Vortex Transformation. One advantage of this approach is that the encrypted images can be used for training machine learning models without the need for decryption. Objectively speaking, this method is not a perfect encryption scheme. Our goal is not to propose the highest level of security encryption scheme but rather to combine data encryption with machine learning. For instance, combining the Random Vortex Transformation with federated learning~\cite{DBLP:conf/aistats/McMahanMRHA17} can better prevent attackers from recovering the original image using gradient attacks and other methods~\cite{DBLP:conf/nips/ZhuLH19, DBLP:journals/pami/GaoZGZXQWL23, DBLP:journals/tifs/YangGXL23}. This is also one of our future directions of application.

\ifCLASSOPTIONcaptionsoff
  \newpage
\fi

\appendices

\bibliographystyle{IEEEtran}

\bibliography{Vortexref}

\vskip -0.3in
\begin{IEEEbiography}
	[{\includegraphics[width=1in,height=1.25in,clip,keepaspectratio]{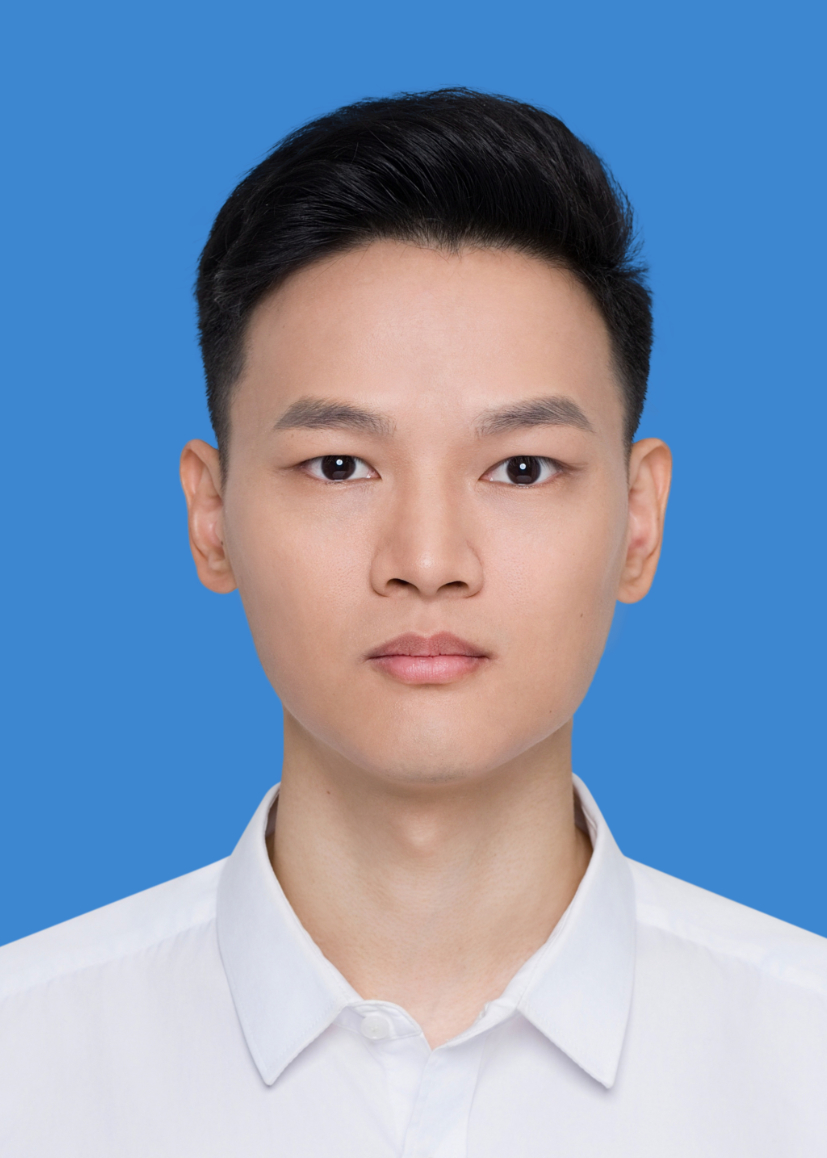}}]
	{Xiao-Kai Cao} received his Master degree in mathematics in 2020 from Guizhou University. He is purchasing his Ph.D. degree in computer science and technology at Sun Yat-sen University. His research interests are privacy computing and information security. He has published several academic papers in journals such as IEEE TCYB.
\end{IEEEbiography}
\begin{IEEEbiography}
    [{\includegraphics[width=1in,height=1.25in,clip,keepaspectratio]{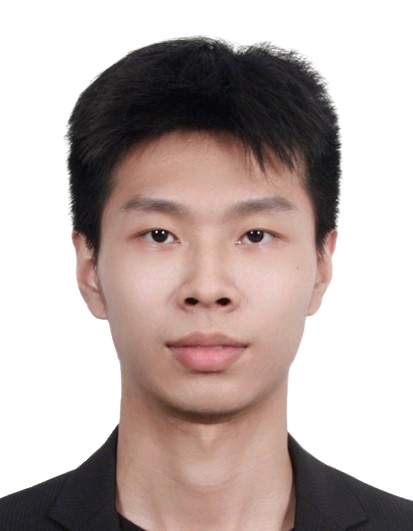}}]{Wen-Jin Mo} is purchasing his bachelor degree in computer science and technology at Sun Yat-sen University. His research interests are computer vison and privacy computing.
\end{IEEEbiography}
\begin{IEEEbiography}
	[{\includegraphics[width=1in,height=1.25in,clip,keepaspectratio]{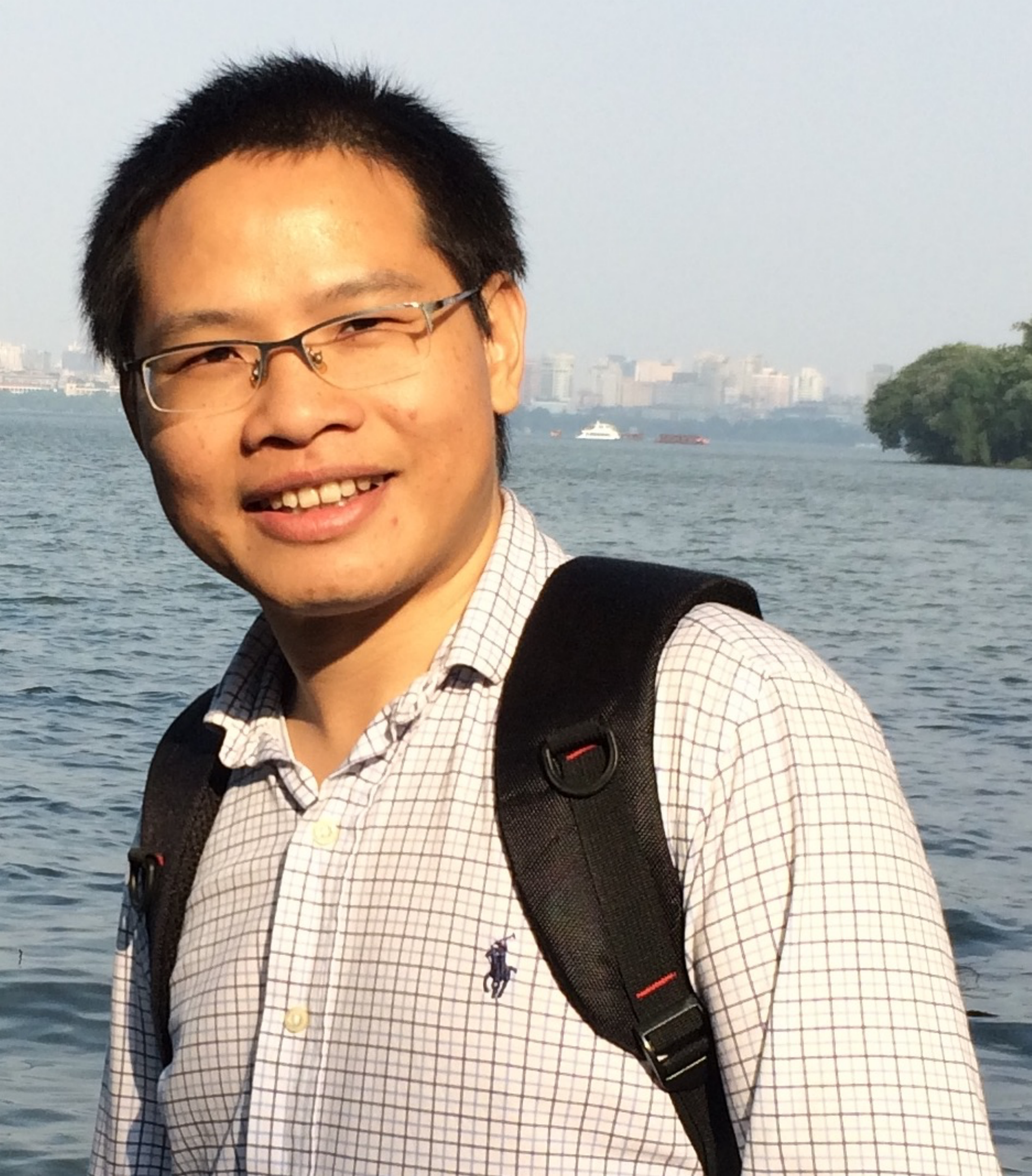}}]{Chang-Dong Wang} received the Ph.D. degree in computer science in 2013 from Sun Yat-sen University, Guangzhou, China. He was a visiting student at University of Illinois at Chicago from Jan. 2012 to Nov. 2012. He joined Sun Yat-sen University in 2013, where he is currently an associate professor with School of Computer Science and Engineering. His current research interests include machine learning and data mining. He has published over 80 scientific papers in international journals and conferences such as IEEE TPAMI, IEEE TKDE, IEEE TCYB, IEEE TNNLS, ACM TKDD, ACM TIST, IEEE TSMC-Systems, IEEE TII, IEEE TSMC-C, KDD, AAAI, IJCAI, CVPR, ICDM, CIKM and SDM. His ICDM 2010 paper won the Honorable Mention for Best Research Paper Awards. He won 2012 Microsoft Research Fellowship Nomination Award. He was awarded 2015 Chinese Association for Artificial Intelligence (CAAI) Outstanding Dissertation. He was an Associate Editor in Journal of Artificial Intelligence Research (JAIR).
\end{IEEEbiography}

\begin{IEEEbiography}
	[{\includegraphics[width=1in,height=1.25in,clip,keepaspectratio]{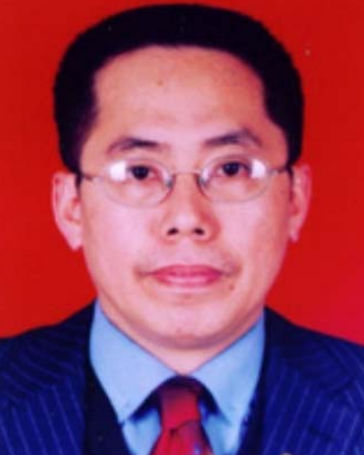}}]{Jian-Huang Lai} (Senior Member, IEEE) received the M.Sc. degree in applied mathematics and the Ph.D. degree in mathematics from Sun Yat-sen University, Guangzhou, China, in 1989 and 1999, respectively.

He joined Sun Yat-sen University, in 1989, as an Assistant Professor, where he is currently a Professor with the School of Data and Computer Science. He has authored or coauthored more than 200 scientific papers in the international journals and conferences on image processing and pattern recognition, such as IEEE TRANSACTIONS ON PATTERN ANALYSIS AND MACHINE INTELLIGENCE, IEEE TRANSACTIONS ON KNOWLEDGE AND DATA ENGINEERING, IEEE TRANSACTIONS ON NEURAL NETWORKS, IEEE TRANSACTIONS ON IMAGE PROCESSING, IEEE TRANSACTIONS ON SYSTEMS, MAN, AND CYBERNETICS-PART B: CYBERNETICS, Pattern Recognition, International Conference on Computer Vision, Conference on Computer Vision and Pattern Recognition, International Joint Conference on Artificial Intelligence, IEEE International Conference on Data Mining, and SIAM International Conference on Data Mining. His research interests include digital image processing, pattern recognition, and multimedia communication, wavelet and its applications. Dr. Lai is as a Standing Member of the Image and Graphics Association of China and also a Standing Director of the Image and Graphics Association of Guangdong.
\end{IEEEbiography}
\begin{IEEEbiography}
	[{\includegraphics[width=1in,height=1.25in,clip,keepaspectratio]{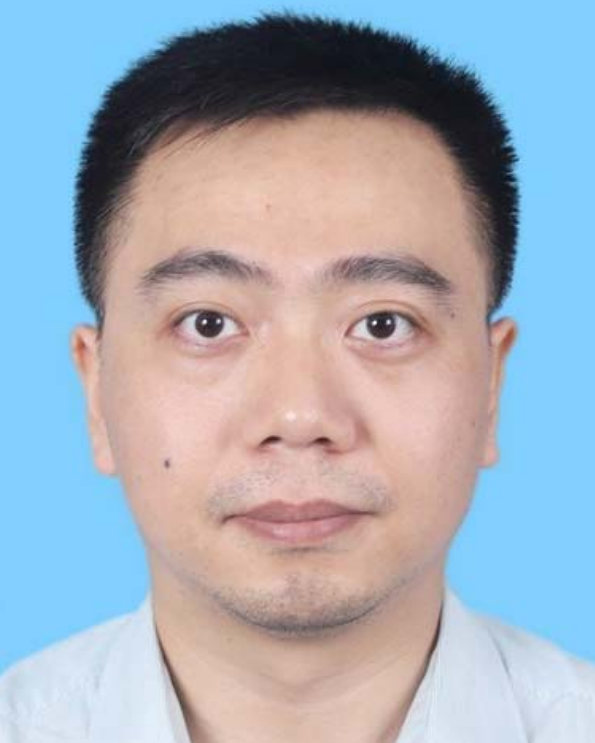}}]{Qiong Huang} (Member, IEEE) received the Ph.D. degree from the City University of Hong Kong in 2010. He is currently a Professor with the College of Mathematics and Informatics, South China Agricultural University, Guangzhou, China. He has authored or coauthored more than 120 research papers in international conferences and journals. His research interests include cryptography and information security, in particular, cryptographic protocols design and analysis. He was a Program Committee Member in many international conferences.
\end{IEEEbiography}
\end{document}